%% file: main.tex
\pdfoutput=1 %uncomment to ensure pdflatex processing (mandatatory e.g. to submit to arXiv)
\documentclass[a4paper, UKenglish, cleveref, autoref, numberwithinsect, pdfa]{lipics-v2021}
\include{settings}
\include{macros}

\begin{document}

\maketitle
\input{sections/abstract}

\input{sections/intro}
\input{sections/preliminaries}

\input{sections/primitive_generators}
\input{sections/upper_bound_AF}
\input{sections/guarded-fragment}

\input{sections/conclusions}

% \clearpage
\bibliographystyle{plain}
\bibliography{references}%
\clearpage
\appendix
\input{sections/appendix-preliminaries}
\input{sections/appendix-upper-bounds}

\end{document}

%% file: settings.tex
\title{On the Limits of Decision:\\ the Adjacent Fragment of First-Order Logic}
\titlerunning{On the Limits of Decision: the Adjacent Fragment of First-Order Logic}

\author{Bartosz Bednarczyk}
{Computational Logic Group, Technische Universit{\"a}t  Dresden, Germany \and 
Institute of Computer Science, University of Wroc\l aw, Poland
\and \url{https://bartoszjanbednarczyk.github.io/}}
{bartosz.bednarczyk@cs.uni.wroc.pl}
{https://orcid.org/0000-0002-8267-7554}
{supported by the ERC Consolidator Grant No.~771779 (\href{https://iccl.inf.tu-dresden.de/web/DeciGUT/en}{DeciGUT}).}

\author{Daumantas Kojelis}
{Department of Computer Science, University of Manchester, UK
	\and \url{https://daumantaskojelis.github.io/}
}
{daumantas.kojelis@manchester.ac.uk}
{https://orcid.org/0000-0002-1632-9498}
{}

\author{Ian Pratt-Hartmann}
{Department of Computer Science, University of Manchester, UK \and
Institute of Computer Science, University of Opole, Poland
\and \url{https://www.cs.man.ac.uk/~ipratt/}
}
{ian.pratt@manchester.ac.uk}
{https://orcid.org/0000-0003-0062-043X}
{supported by the NCN grant 2018/31/B/ST6/03662.}

\authorrunning{B. Bednarczyk, D. Kojelis and I. Pratt-Hartmann}
\Copyright{Bartosz Bednarczyk, Daumantas Kojelis and Ian Pratt-Hartmann}

\ccsdesc[500]{Theory of computation~Finite Model Theory}
\keywords{
decidability,
satisfiability,
variable-ordered logics,
complexity
}

\relatedversion{Full version: \href{https://arxiv.org/abs/2305.03133}{https://arxiv.org/abs/2305.03133}~\cite{bednarczyk2023limits}.}
\nolinenumbers %uncomment to disable line numbering
\acknowledgements{B. Bednarczyk thanks R. Jaakkola for many inspiring discussions.}

\EventEditors{Kousha Etessami, Uriel Feige, and Gabriele Puppis}
\EventNoEds{3}
\EventLongTitle{50th International Colloquium on Automata, Languages, and Programming (ICALP 2023)}
\EventShortTitle{ICALP 2023}
\EventAcronym{ICALP}
\EventYear{2023}
\EventDate{July 10--14, 2023}
\EventLocation{Paderborn, Germany}
\EventLogo{}
\SeriesVolume{261}
\ArticleNo{114}
\usepackage{microtype} % Micro typing improvements
\usepackage{ellipsis} % improves whitespacing around "..."s
\usepackage[abbreviations,british]{foreign} % for better typing i.e., e.g. etc. 

\usepackage{xcolor}
\usepackage{hyperref}

\definecolor{darkmidnightblue}{rgb}{0.0, 0.2, 0.4}
\definecolor{persianplum}{rgb}{0.44, 0.11, 0.11}

\hypersetup{
  colorlinks  = true, % Colours links instead of ugly boxes
  citecolor   = persianplum, % Colour of citations
  urlcolor    = persianplum, % Colour for external hyperlinks,
  linkcolor   = persianplum
}

\usepackage{tikz}
\usepackage{todonotes}
\usepackage{xspace}

\usepackage{mathtools}
\usepackage{amssymb} 
\usepackage{amsmath}
\usepackage{amsthm}

\usepackage{caption}
\usepackage{float}
\usepackage{subcaption}

\usepackage[h]{esvect}

\makeatletter
\def\desclabel#1#2{\begingroup
\def\@currentlabel{#1}%
#1\label{#2}\endgroup
}
\makeatother

\usetikzlibrary{shadows}
\tikzset{
diagonal fill/.style 2 args={fill=#2, path picture={
\fill[#1, sharp corners] (path picture bounding box.south west) -|
                         (path picture bounding box.north east) -- cycle;}},
reversed diagonal fill/.style 2 args={fill=#2, path picture={
\fill[#1, sharp corners] (path picture bounding box.north west) |- 
                         (path picture bounding box.south east) -- cycle;}}
}

\usetikzlibrary{hobby,backgrounds,calc,trees}

\usepackage{tikz-cd}
\usetikzlibrary{matrix}

\usepackage{stackengine}
\usepackage{xspace}
\usepackage{xcolor}
\usepackage[linesnumbered,ruled,vlined]{algorithm2e}

\SetCommentSty{mycommfont}
\SetKwInput{KwData}{Input}

\usetikzlibrary{intersections}
\usetikzlibrary{decorations}
\usetikzlibrary{snakes}
\usepackage{enumitem}
\newtheorem*{theorem*}{Theorem}
\newtheorem*{lemma*}{Lemma}
\usepackage{tikz}
\usetikzlibrary{calc,arrows,decorations.pathmorphing,decorations.markings,decorations.pathreplacing,snakes}
\tikzstyle{element}=[circle,draw=blue,fill=blue,inner sep=0pt,minimum size=3]

%% file: macros.tex
\tikzset{->-/.style={decoration={
			markings,
			mark=at position #1 with {\arrow{>}}},postaction={decorate}}}
\tikzset{-<-/.style={decoration={
			markings,
			mark=at position #1 with {\arrow{<}}},postaction={decorate}}}

\tikzset{-|>-/.style={decoration={
			markings,
			mark=at position #1 with {\arrow{open triangle 60}}},postaction={decorate}}}

\definecolor{ao(english)}{rgb}{0.0, 0.5, 0.0}

\renewcommand{\phi}{\varphi}

\newcommand{\logic}[1]{\mathsf{#1}}

\newcommand{\FO}{\ensuremath{\logic{FO}}}
\newcommand{\FOt}{\FO^{2}}

\newcommand{\FL}{\ensuremath{\mathcal{FL}}}          % Fluted fragment
\newcommand{\FLv}[1]{\ensuremath{\mathcal{FL}^{#1}}} % ... superscripted
\newcommand{\AF}{\ensuremath{\mathcal{AF}}}          % Fluted fragment
\newcommand{\AFv}[1]{\ensuremath{\mathcal{AF}^{#1}}} % ... superscripted
\newcommand{\GA}{\ensuremath{\mathcal{GA}}} % Guarded adjacent fragment
\newcommand{\GF}{\ensuremath{\mathcal{GF}}} % Guarded fragment

\newcommand{\complexityclass}[1]{\textsc{#1}} % any complexity class
\newcommand{\NP}{\complexityclass{NPTime}}

\newcommand{\PSpace}{\complexityclass{PSpace}}

\newcommand{\AExpSpace}{\complexityclass{AExpSpace}}
\newcommand{\NExpTime}{\complexityclass{NExpTime}} % nondeterministic exponential time
\newcommand{\TwoExpTime}{\complexityclass{2ExpTime}} % doubly-exponential time
\newcommand{\Tower}{\complexityclass{Tower}} %

\newcommand{\set}[1]{\ensuremath{\{ #1 \}}} % Notation for sets
\newcommand{\sizeof}[1]{|\!|#1|\!|}         % Size

\newcommand{\bA}{\ensuremath{\mathbf{A}}} % Sets of adjacent functions
\newcommand{\bx}{\ensuremath{\mathbf{x}}} % The variables x_1, ..., x_k
\newcommand{\bu}{\ensuremath{\mathbf{z}}} % The variables z_1, ..., z_k
\newcommand{\bv}{\ensuremath{\mathbf{z}'}} % The variables z'_1, ..., z'_k

\newcommand{\N}{\ensuremath{\mathbb{N}}} % Natural numbers
\newcommand{\fA}{\ensuremath{\mathfrak{A}}} % Structure
\newcommand{\fB}{\ensuremath{\mathfrak{B}}} % Structure
\newcommand{\fC}{\ensuremath{\mathfrak{C}}} % Structure

\newcommand{\ft}{\ensuremath{\mathfrak{t}}} % tower function

\newcommand{\Atp}{\ensuremath{\mbox{\rm Atp}}}  % Set of adjacent types
\newcommand{\atp}{\ensuremath{\mbox{\rm atp}}}  % Adjacent types
\newcommand{\tp}{\ensuremath{\mbox{\rm tp}}}  % Adjacent types
\newcommand{\con}{\ensuremath{\mbox{\rm con}}}  % Connector types

\newcommand{\blank}{\ensuremath{\llcorner \hspace{-1mm} \lrcorner}} % blank symbol 

\newcommand{\val}[1]{\ensuremath{\textsc{val}^{#1}}} % value function (bit strings -> naturals)
\newcommand{\bin}[2]{\ensuremath{\textsc{bin}^{#1}_{#2}}} % value function (naturals -> bin in x, y)

\newcommand{\tm}{\ensuremath{\mathcal{M}}}  % Turing Machine
\newcommand{\tmst}{\ensuremath{\mathcal{Q}}}  % Turing Machine States
\newcommand{\tmal}{\ensuremath{\mathcal{S}}}  % Turing Machine Alphabet
\newcommand{\tmtrl}{\ensuremath{\mathcal{T}_l}}  % Turing Machine Transition Left
\newcommand{\tmtrr}{\ensuremath{\mathcal{T}_r}}  % Turing Machine Transition Right
\newcommand{\tmtra}{\ensuremath{\mathcal{T}_\eta}}  % Turing Machine Transition (arbitrary)

\newcommand{\tmuni}{\ensuremath{\pmb{\forall}}}  % Turing Machine Universal State
\newcommand{\tmexi}{\ensuremath{\pmb{\exists}}}  % Turing Machine Existential State
\newcommand{\tmrej}{\ensuremath{\pmb{\times}}}  % Turing Machine Rejecting State
\newcommand{\tmct}{\ensuremath{\mathcal{G}}}  % Turing Machine Configuration Tree
\newcommand{\tmctver}{\ensuremath{\mathcal{V}}}  % Turing Machine Configuration Tree Vertices
\newcommand{\tmctverl}{\ensuremath{\mathcal{E}_l}}  % Turing Machine Configuration Tree Left Edge
\newcommand{\tmctverr}{\ensuremath{\mathcal{E}_r}}  % Turing Machine Configuration Tree Right Edge
\newcommand{\tmctvera}{\ensuremath{\mathcal{E}_\eta}}  % Turing Machine Configuration Tree Left (arbitrary)

\newcommand{\sheq}[2]{\ensuremath{\textsc{eq}(#1, #2)}} % equals short-hand
\newcommand{\shless}[2]{\ensuremath{\textsc{less}(#1, #2)}} % less short-hand

\newcommand{\str}[1]{{\mathfrak{#1}}}

\newcommand{\sig}{\mathsf{sig}} % signatures

\newcommand{\bfz}{\ensuremath{\mathbf{0}}} % bold 0
\newcommand{\bfo}{\ensuremath{\mathbf{1}}} % bold 1

%% file: sections/abstract.tex
%!TEX root = ../main.tex

\begin{abstract}
We define the \emph{adjacent fragment} $\AF$ of first-order logic, obtained by restricting the sequences of variables occurring as arguments in atomic formulas. 
The adjacent fragment generalizes (after a routine renaming) two-variable logic as well as the fluted fragment.
We show that the adjacent fragment has the finite model property, and that its satisfiability problem is no harder than for the fluted fragment (and hence is $\Tower$-complete). 
We further show that any relaxation of the adjacency condition on the allowed order of variables in
argument sequences yields a logic whose satisfiability and finite satisfiability problems are undecidable.
Finally, we study the effect of the adjacency requirement on the well-known guarded fragment ($\GF$) of first-order logic.
We show that the satisfiability problem for the guarded adjacent fragment ($\GA$) remains $\TwoExpTime$-hard, 
thus strengthening the known lower bound for $\GF$.
\end{abstract}

%% file: sections/intro.tex
%!TEX root = ../main.tex
\section{Introduction}
\label{sec:intro}

The quest to find fragments of first-order logic for which satisfiability is algorithmically decidable has been a central undertaking of mathematical logic since  the 
appearance of Hilbert and Ackermann's {\em Grundz\"{u}ge der theoretischen Logik}~\cite{hilbert, book:ha50} almost a century ago. The great
majority of such fragments so far discovered, however, belong to just three families:
(i) quantifier prefix fragments~\cite{BorgerGG1997}, where we are restricted to prenex formulas with a specified quantifier sequence; 
(ii) two-variable logics~\cite{Henkin1967}, where the only logical variables occurring as arguments of predicates are $x_1$ and $x_2$; and (iii)
guarded logics~\cite{ABN98}, where quantifiers are relativized by atomic formulas featuring all the free variables in their scope.

There is, however, a fourth family of decidable logics, originating in the work of\linebreak
 W.V.O.~Quine~\cite{quine69}, and based on restricting the allowed sequences of variables occurring as arguments in atomic formulas. 
This family of logics, which includes the {\em fluted fragment}, the {\em ordered fragment} and the 
{\em forward fragment}, has languished in relative obscurity. In this paper, we investigate the potential for obtaining decidable fragments in this way, identifying a new fragment, which we call the {\em adjacent fragment}. This fragment not only includes the fluted, ordered and forward fragments, but also subsumes, in a sense we make precise, the two-variable fragment. We show that the satisfiability problem for the adjacent fragment is decidable, and determine bounds on its complexity.
 
To explain how restrictions on argument orderings work, we consider presentations of first-order logic without equality over purely relational signatures, employing individual variables from the alphabet $\set{x_1, x_2, x_3, \dots}$. Any atomic formula in this logic has the form~$p(\bar{x})$, where~$p$ is a predicate of arity $m \geq 0$ and $\bar{x}$ is a word over the alphabet of variables of length $m$. Call a first-order formula $\phi$ {\em index-normal} if, for any quantified sub-formula $Q x_{k}\, \psi$ of $\phi$, $\psi$ is 
a Boolean combination of formulas that are either atomic with free variables among $x_1, \dots, x_k$,
or have as their major connective a quantifier binding $x_{k+1}$. By re-indexing variables, any first-order formula can easily be written as a
logically equivalent index-normal formula. In the {\em fluted fragment}, denoted \FL, as defined by W.~Purdy~\cite{purdy96}, we confine attention to index-normal formulas, but additionally insist that any atom occurring in a context in which~$x_k$ is quantified have the form $p(x_{k - m +1} \cdots x_k)$, i.e.~$p(\bar{x})$ with $\bar{x}$ a {\em suffix} of
$x_1 \cdots x_k$. 
In the {\em ordered fragment}, due to A.~Herzig~\cite{herzig90}, by contrast,
we insist that $\bar{x}$ be a {\em prefix} of
$x_1 \cdots x_k$.  In the \emph{forward fragment}~\cite{Bednarczyk21}, 
we insist only that $\bar{x}$ be an {\em infix} of
$x_1 \cdots x_k$. 

All these logics have the finite model property, and hence are decidable for  satisfiability.
Denoting by $\FLv{k}$ the sub-fragment of \FL{} involving at most $k$ variables (free or bound), 
the satisfiability problem for $\FLv{k}$ is known to be in $(k{-}2)$-\NExpTime{} for
all $k \geq 3$, and $\lfloor k/2\rfloor$-\NExpTime-hard for all $k \geq 2$~\cite{phst19}.
Thus, satisfiability for the whole fluted fragment is \Tower-complete, in the system of trans-elementary complexity classes due to~\cite{schmitz16}.
By contrast, the satisfiability problem for the ordered fragment is known to be $\PSpace$-complete~\cite{herzig90}\cite{Jaakkola21}.
On the other hand, the apparent liberalization afforded by the forward fragment yields no difference in expressive power~\cite{BednarczykJ22}, and moreover
there is a polynomial time satisfiability-preserving reduction of the forward fragment to the fluted fragment~\cite{Bednarczyk21}.
The term ``fluted'' originates with Quine~\cite{quine76a}, and presumably invites us to imagine the atoms in formulas aligned in such a way that the variables form columns. 
%Analysis of Quine's original text suggests that he may in fact have had in mind what we are calling the \textit{ordered} fragment. However, following the work of Purdy, the term ``fluted fragment'' has now become established in the sense indicated above, and we shall not attempt to right any historical wrongs. 
Note that none of these fragments can state that a relation is reflexive or symmetric (see~\cite{BednarczykJ22} for a discussion~of~their~expressivity).

Say that a word $\bar{x}$ over the alphabet $\set{x_1, \dots, x_k}$ ($k \geq 0$) is \emph{adjacent} if the indices of neighbouring letters differ by at most 1. For example,
$x_3x_2x_1x_2x_2x_2x_3x_4x_3$ is adjacent, but $x_1x_3x_2$ is not.  The \textit{adjacent fragment}, denoted~\AF, 
is analogous to the fluted, ordered and forward fragments, but we allow any atom $p(\bar{x})$
to occur in a context where $x_k$ is available for quantification as long as $\bar{x}$ is an adjacent word over $\set{x_1, \dots, x_k}$. (A formal definition is given in Sec.~\ref{sec:preliminaries}.) 
As a simple example, the formula 
\begin{equation}
	\forall{x_1}\forall{x_2}\forall{x_3}\exists{x_4}\forall{x_5} \ \big(p(x_1 x_2 x_3 x_2 x_3 x_4 x_5) \to p(x_1 x_2 x_3 x_4 x_3x_4x_5)\big) 
\label{eq:simpleExample}
\end{equation}
is a validity of \AF, as can be seen by assigning $x_4$ the same value as $x_2$. 
Evidently, \AF{} includes the fluted, ordered and 
forward fragments; the inclusion is strict, since the
formulas $\forall x_1\, r(x_1x_1)$ and 
$\forall x_1 x_2 (r(x_1x_2)\rightarrow r(x_2x_1))$, stating that $r$ is reflexive and symmetric,
respectively, are in \AF. 
As every word over~$\{ x_1, x_2 \}$ is adjacent, we may transform any formula of the two-variable fragment without equality,  $\FOt$,
in polynomial time, to a logically equivalent formula of $\AF$.
The converse is true over signatures with predicates of arity at most two.
Since the system of basic multimodal propositional logic is, under the standard translation to first-order logic, included within
$\FOt$, this logic is similarly subsumed by \AF, as indeed
is its notational variant, the description logic $\mathcal{ALC}$ (see, e.g.~\cite{hsg04}). 

We show that the satisfiability problem
for the restriction of the adjacent fragment to formulas involving at most $k$ 
variables (free or bound) is in $(k{-}2)$-\NExpTime{} for all $k \geq 3$---and hence no more difficult than the $k$-variable fluted fragment, which it properly contains.
The critical step in our analysis is a lemma on the combinatorics of strings (Theorem~\ref{theo:main}), which may be of independent interest.
We also consider minimal relaxations of
adjacency involving the fragment with just three variables, and show that, in all cases of interest,
the satisfiability and finite satisfiability problems for the resulting logics are undecidable. Thus, adjacency is
as far as we can go in seeking decidable fragments based on straightforward argument ordering restrictions of the type envisaged by Quine. 

The
adjacent fragment is incomparable in expressive power to the guarded fragment.
Moreover, the satisfiability problem for the union of $\GF$ and $\AF$ is undecidable, 
as one can use adjacent formulas to introduce any $k$-ary universal relations, 
which makes $\GF$ as expressive as first-order logic.
Therefore, we study the effect of the adjacency restriction~on~$\GF$.
We investigate the complexity of satisfiability for the resulting logic, $\GA$, showing that the problem 
is \TwoExpTime-complete, thus sharpening the existing $\TwoExpTime$-hardness 
proof for~$\GF$~\cite{Gradel99}.

%% file: sections/preliminaries.tex
%!TEX root = ../main.tex

\section{Preliminaries}
\label{sec:preliminaries}
%We take the {\em natural numbers}, $\N$, to be the set $\set{0, 1, 2, \dots}$. 
Let $m$ and $k$ be non-negative integers. 
For any integers $i$ and $j$, we write $[i,j]$ to denote the set of integers $h$ such that $i \leq h \leq j$.  
A function $f \colon [1,m] \rightarrow [1,k]$ is {\em adjacent}
if $|f(i+1) {-} f(i)| \leq 1$ for all $i$ ($ 1 \leq i < m$). 
% Equivalently, 
% $f$ is adjacent if
% $|f(i) {-} f(j)| \leq |i {-} j|$ for all $i, j \in [1,m]$, that is to say, if $f$ is a metric map under the standard metrics on its domain and co-domain. 
We write $\bA^m_k$ to denote the set of adjacent functions $f \colon [1,m] \rightarrow [1,k]$. Since $[1,0] = \emptyset$, we have $\bA^0_k = \set{\emptyset}$, and $\bA^m_0 = \emptyset$ if $m >0$.
Let $A$ be a non-empty set. A word $\bar{a}$ over the alphabet $A$ is 
simply a tuple of elements from $A$; we alternate freely in the sequel between these two ways of
speaking as the context requires. Accordingly, $A^k$ denotes the set of words 
over $A$ having length exactly $k$, and $A^*$ is the set of all finite words over $A$. % having any finite length (possibly zero).
If $\bar{a} \in A^*$, denote the length of $\bar{a}$ by $|\bar{a}|$, and the reversal of $\bar{a}$  by $\bar{a}^{-1}$. 
Any function
$f\colon [1,m] \rightarrow [1,k]$ (adjacent or not) induces a natural map from $A^k$ to $A^m$ defined by
$\bar{a}^f = a_{f(1)} \cdots a_{f(m)}$, where $\bar{a} = a_1 \cdots a_k$.
If $f \in \bA^m_k$ (i.e.~if $f$ is adjacent), we may think of $\bar{a}^f$ as the result of a `walk' on the tuple $\bar{a}$, starting at the element $a_{f(1)}$, and moving left, right,
or remaining stationary according to the sequence of values $f(i+1) {-} f(i)$ ($1 \leq i <m$).

For any $k \geq 0$, denote by $\bx_k$ the fixed word $x_1 \cdots x_k$ (if $k=0$, this is the empty word). 
A {\em $k$-atom} is an expression $p(\bx^f_k)$, where $p$ is a predicate of 
some arity $m \geq 0$, and $f\colon [1,m] \rightarrow [1,k]$. Thus,
in a $k$-atom, each argument is a variable chosen from $\bx_k$. If $f$ is adjacent, we speak of an
{\em adjacent $k$-atom}. Thus,
in an adjacent $k$-atom, the indices of neighbouring
arguments differ by at most one. When $k \leq 2$, the adjacency requirement is vacuous, and in this case we prefer to speak 
simply of \textit{$k$-atoms}. 
Proposition letters (predicates of arity $m=0$) count as (adjacent) $k$-atoms 
for all $k \geq 0$, taking $f$ to be the empty function. 
When $k= 0$, we perforce have $m = 0$, since otherwise, there are no functions from $[1,m]$ to~$[1,k]$; thus the 0-atoms are 
precisely the proposition letters.

We define the sets of first-order formulas $\AFv{[k]}$ by simultaneous structural induction:
\begin{enumerate}
  \item every adjacent $k$-atom is in $\AF^{[k]}$; 
  \item $\AF^{[k]}$ is closed under Boolean combinations;
  \item if $\phi$ is in $\AFv{[k+1]}$,  $\exists x_{k+1}\, \phi$ and  $\forall x_{k+1}\, \phi$ are in $\AFv{[k]}$. 
\end{enumerate}
Now let $\AF=\bigcup_{k \geq 0} \AF^{[k]}$ and define $\AFv{k}$ to be the set of formulas of $\AF$ featuring no variables other than $\bx_k$, free or bound. 
We call
$\AF$ the {\em adjacent fragment} and $\AFv{k}$ the {\em $k$-variable adjacent fragment}. 
Note that formulas of $\AF$ contain no individual constants, function symbols or equality.
The primary objects of interest here are the languages  $\AF$ and~$\AFv{k}$; however, the sets of formulas $\AF^{[k]}$ play an important
auxiliary role in their analysis. 
Thus, for example, formula~\eqref{eq:simpleExample} is in $\AFv{k}$ for all $k \geq 5$, but in $\AFv{[k]}$ only for $k =0$.
On the other hand, the \textit{quantifier-free} formulas of 
$\AFv{[k]}$ and $\AFv{k}$ are the same.

We silently assume the variables $\bx_k = x_1 \cdots x_k$ to be ordered in the standard way. That is: if 
$\phi$ is a formula of $\AF^{[k]}$, $\fA$ a structure interpreting its signature, and $\bar{a} = a_1 \cdots a_k \in A^k$, 
we say simply that $\bar{a}$ \textit{satisfies} $\phi$ in $\fA$, and write $\str{A} \models \phi[\bar{a}]$ to mean that 
$\bar{a}$ satisfies~$\phi$ in $\fA$ under the assignment $x_i \leftarrow a_i$ ($1 \leq i \leq k)$. (This does not 
necessarily mean that 
each of the variables of $\bx_k$ actually appears in $\phi$.) 
If $\phi$ is true under  all assignments in all structures,
we write $\models \phi$; the notation
$\phi \models \psi$ means the same as $\models \phi \rightarrow \psi$ (i.e.~variables are consistently instantiated in $\phi$ and $\psi$).
%We use $\sig(\varphi)$ to denote the set of predicates used in $\varphi$.\dauside{Migth be a good idea to say that $\sig$ is the signature explicitly, as we sometimes say "Let $\sigma$ be the signature..."}
The notation $\phi(\bar{v})$, where $\bar{v}=v_1 \cdots v_k$ are variables, will always be used to denote
the formula that results from substituting $v_i$ for $x_i$ ($1 \leq i \leq k)$~in~$\phi$. 
We write $\forall \bx_{k;\ell}$ in place of $\forall x_k \forall x_{k{+}1} \cdots \forall x_\ell$ (and just $\forall \bx_{\ell}$~if~$k=1$).
A~{\em sentence} is a formula with no free variables. Necessarily, all formulas 
of $\AFv{[0]}$ are sentences.
For a sentence $\varphi$ we write simply $\fA \models \phi$ to mean that $\phi$ is true in $\fA$.
% If $\phi$ is a \textit{sentence} of $\AFv{[k]}$ for some $k >0$---for example, $\exists x_{2}\, p(x_2)$ is in $\AFv{[1]}$---we continue to write simply $\fA \models \phi$ to mean that $\phi$ is true in $\fA$.
%
We call the set of predicates used in $\varphi$ \emph{the signature of $\varphi$} (denoted~$\sig(\varphi)$).
%
% In the sequel all signatures will be silently assumed to be purely relational. 
By routine renaming of variables we establish:
\begin{lemma}\label{lemma:fo2-and-af-over-binary-sig-are-the-same}
Every $\FO^2$-formula is logically equivalent to an~$\AF$-formula. 
The converse holds for $\AF$-formulas featuring predicates of arity at most two.
\end{lemma}

We adapt the standard notion of (atomic) $k$-types for the fragments studied here. Fix some signature $\sigma$. 
An {\em adjacent $k$-literal} \textit{over} $\sigma$ is an {\em adjacent $k$-atom} or its negation, featuring a predicate in $\sigma$.
An \emph{adjacent $k$-type} \textit{over} $\sigma$ is a maximal consistent set of adjacent $k$-literals over $\sigma$. 
Reference to $\sigma$ is suppressed where clear from context. 
We use the letters $\zeta$ and~$\eta$ to range over 
adjacent $k$-types for various $k$. We denote by
$\Atp^\sigma_k$ the set of all adjacent $k$-types over $\sigma$. For finite $\sigma$, we identify members of
$\Atp^\sigma_k$ with their conjunctions, and treat them as (quantifier-free) $\AFv{k}$-formulas, 
writing $\zeta$ instead of $\bigwedge \zeta$. When $k \leq 2$, the adjacency requirement is vacuous, and in this case we shall simply speak simply of \textit{$k$-types}. It is obvious that every quantifier-free $\AFv{k}$-formula $\chi$
is logically equivalent to a disjunction of adjacent $k$-types, in essence the adjacent disjunctive normal form of $\chi$. In particular,
if $\chi$ is satisfiable, then there is an adjacent $k$-type which entails it.
If $\fA$ is a $\sigma$-structure and $\bar{a}$ a $k$-tuple of elements from $A$, there 
is a unique adjacent $k$-type $\zeta$ such that $\fA \models \zeta[\bar{a}]$; we denote
this adjacent $k$-type by $\atp^\fA[\bar{a}]$, and call it the {\em adjacent type of $\bar{a}$ in $\fA$}.
If $\tau \subseteq \sigma$, we use $\atp^\fA_{\tau}[\bar{a}]$ to denote the adjacent type of $\bar{a}$ in $\fA$ restricted to predicates from $\tau$.
It is not required that the elements $\bar{a}$ be distinct. Again, if $k \leq 2$, adjacency is vacuous, and we 
write $\tp^\fA[\bar{a}]$ rather than $\atp^\fA[\bar{a}]$, and refer to $\tp^\fA[\bar{a}]$ as the {\em type of $\bar{a}$ in $\fA$}.

Since adjacent formulas do not contain equality, we may freely duplicate elements in their models. 
Let $\fB$ be a $\sigma$-structure, and $I$ a non-empty set of indices. We define 
the structure $\fB \times I$ over the Cartesian product $B \times I$ by setting, for any
$p \in \sigma$ of arity $m$, and any $m$-tuples $b_1 \cdots b_m$ from $B$ and $i_1 \cdots i_m$ from $I^m$,
$\fB \times I \models p[\langle b_1,i_1\rangle  \cdots \langle b_m, i_m\rangle]$ if and only if 
$\fB \models p[b_1 \cdots b_m]$.  By routine structural induction:
\begin{lemma}
  Let $\psi$ be an equality-free first-order formula all of whose free variables occur in $\bx_k$, $\fB$
  a structure interpreting the signature of $\psi$, and $I$ a non-empty set. Then, for any tuples 
  $\bar{b} = b_1 \cdots b_k$ from $B$ and $i_1 \cdots i_k$ from $I$,
  $\fB \models \psi[\bar{b}]$ if and only if $\fB \times I \models \psi[\langle b_1, i_1\rangle \cdots \langle b_k, i_k\rangle]$.\label{lma:cartesian}
\end{lemma}
The following combinatorial lemma allows us to extend the technique of `circular witnessing'~\cite{GradelKV97}, 
frequently used in the analysis of two-variable logics, to the languages $\AFv{k}$.
\begin{lemma}\label{lma:simpleComb} 
For any integer $k>0$ there is a set $J$ with $|J| = (k^2 + k+1)^{k+1}$ 
and a function $g\colon J^k \rightarrow J$ such that,
for any tuple $\bar{t} \in J^k$ consisting of the elements $t_1, \dots, t_k$ in some order:
{\em (i)} $g(\bar{t})$ is not in $\bar{t}$;
{\em (ii)} if $\bar{t}' \in J^k$ consists of the elements $\set{t_2, \dots, t_k, g(\bar{t})}$ in some order, then
$g(\bar{t}')$  is not in $\bar{t}$ either.
\end{lemma}

%% file: sections/primitive_generators.tex
%!TEX root = ../main.tex

\section{Primitive generators of words}
\label{sec:primGen}
The upper complexity bounds obtained below depend on an observation concerning the combinatorics of words, which may be of independent interest.
For words $\bar{a}, \bar{c} \in A^*$ %are words over %some alphabet 
with $|\bar{a}|= k$ and $|\bar{c}|= m$, 
say that $\bar{a}$ {\em generates} $\bar{c}$ if $\bar{c} = \bar{a}^f$ for some \textit{surjective} function $f \in \bA^m_k$. 
As explained above, it helps to think of $\bar{a}^f$ as the sequence of letters encountered on an $m$-step `walk' backwards and forwards on the tuple $\bar{a}$, with $f(i)$ giving the index of 
our position in $\bar{a}$ at the $i$th step. The condition that 
$f$ is adjacent ensures that we never change position by more than one letter at a time; the condition that $f$ is surjective ensures that we visit every position of $\bar{a}$. 
We may picture a walk as a piecewise linear function,
with the \mbox{generat{\em{}ed}} word superimposed on the abscissa and the \mbox{generat{\em ing}} word on the ordinate, c.f.  Figure~\ref{fig:example}.
% For example, Fig.~\ref{fig:example} shows how
% $\bar{a} =$ {cbadefba} generates $\bar{c}=$ {abcbaaadefedadefbabf}. 
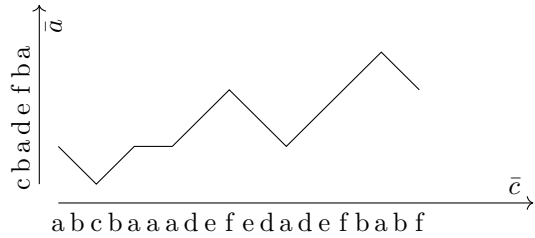
\begin{figure}[h]
  \begin{center}
    \begin{tikzpicture}[scale= 0.25]
      \draw[->] (0,0) --  (25,0);
      \coordinate[label={$\bar{c}$}] (wLabel) at (24,0);
      
      \draw[->-=1] (-1,1) to (-1, 10.5);
      \coordinate[label={\rotatebox{90}{$\bar{a}$}}] (wLabel) at(-0.25, 8.5) ;
      
      \coordinate[label=left:{\rotatebox{90}{c}}] (wLabel) at (-1,1);
      \coordinate[label=left:{\rotatebox{90}{b}}] (wLabel) at (-1,2);
      \coordinate[label=left:{\rotatebox{90}{a}}] (wLabel) at (-1,3);
      \coordinate[label=left:{\rotatebox{90}{d}}] (wLabel) at (-1,4);
      \coordinate[label=left:{\rotatebox{90}{e}}] (wLabel) at (-1,5);
      \coordinate[label=left:{\rotatebox{90}{f}}] (wLabel) at (-1,6);
      \coordinate[label=left:{\rotatebox{90}{b}}] (wLabel) at (-1,7);
      \coordinate[label=left:{\rotatebox{90}{a}}] (wLabel) at (-1,8);
      
      \coordinate[label={a}] (wLabel) at (0,-2);
      \coordinate[label={b}] (wLabel) at (1,-2);
      \coordinate[label={c}] (wLabel) at (2,-2);
      \coordinate[label={b}] (wLabel) at (3,-2);
      \coordinate[label={a}] (wLabel) at (4,-2);
      \coordinate[label={a}] (wLabel) at (5,-2);
      \coordinate[label={a}] (wLabel) at (6,-2);
      \coordinate[label={d}] (wLabel) at (7,-2);
      \coordinate[label={e}] (wLabel) at (8,-2);
      \coordinate[label={f}] (wLabel) at (9,-2);
      \coordinate[label={e}] (wLabel) at (10,-2);
      \coordinate[label={d}] (wLabel) at (11,-2);
      \coordinate[label={a}] (wLabel) at (12,-2);
      \coordinate[label={d}] (wLabel) at (13,-2);
      \coordinate[label={e}] (wLabel) at (14,-2);
      \coordinate[label={f}] (wLabel) at (15,-2);
      \coordinate[label={b}] (wLabel) at (16,-2);
      \coordinate[label={a}] (wLabel) at (17,-2);
      \coordinate[label={b}] (wLabel) at (18,-2);
      \coordinate[label={f}] (wLabel) at (19,-2);
      
      \draw (0,3) -- (2,1) -- (4,3) -- (6,3) -- (9,6) -- (12,3) --(17,8) -- (19,6);
    \end{tikzpicture}
  \end{center}
  \caption{Generation of {abcbaaadefedadefbabf} from {cbadefba}.}
  \label{fig:example}
\end{figure}

Every word generates both itself and its reversal. Moreover, if $\bar{a}$ generates $\bar{c}$, then
$|\bar{c}| \geq |\bar{a}|$; in fact, $\bar{a}$ and $\bar{a}^{-1}$ are the only words of length $|\bar{a}|$ generated by $\bar{a}$. 
Finally, generation is transitive: if $\bar{a}$ generates $\bar{b}$ and $\bar{b}$ generates $\bar{c}$, then $\bar{a}$ generates $\bar{c}$.
We call $\bar{a}$ \textit{primitive} if it is not generated by any word shorter than itself, equivalently, if it is generated only by itself and its reversal. 
For example, $babcd$ and $abcbcd$ are not primitive, because they are generated by $abcd$; but
$abcbda$ is primitive.  
Note that an infix (factor) of a primitive word need not be primitive.
Define  a {\em primitive generator} of $\bar{c}$ to be
a generator of $\bar{c}$ that is itself primitive.
From the foregoing remarks, it is obvious that every word $\bar{c}$ has some primitive generator $\bar{a}$, and indeed, $\bar{a}^{-1}$ as well,
since the reversal of a primitive generator is clearly a primitive generator. The following result, by contrast,
is anything~but~obvious. Notwithstanding the naturalness of the question it answers, we believe it to be new. Since it is concerned only with the
combinatorics of strings, however, we refer the reader to~\cite{ph22} for the proof. 
\begin{theorem}
  The primitive generator of any word is unique up to reversal.
  \label{theo:main}
\end{theorem}
Remarkably, while primitive generators are unique up to reversal, modes of generation are not. The word $\bar{c}= abcbcbd$ has primitive generator
$\bar{a}= abcbd$. But there are \textit{distinct} surjective functions $f, g \in \bA^7_5$ such that $\bar{c}= \bar{a}^f = \bar{a}^g$, as is easily verified.
Define the {\em primitive length} of any word $\bar{c}$ to be the length of any primitive generator of $\bar{c}$. By Theorem~\ref{theo:main},
this notion is well-defined; it will play a significant role in our analysis of the adjacent fragment.
Obviously, the primitive length of $\bar{c}$ is at most $|\bar{c}|$, but will be strictly less if $\bar{c}$ is not primitive. 

Let $\chi$ be a quantifier-free $\AFv{\ell}$-formula, and let $g \in \bA^\ell_k$. We denote by $\chi^g$ the formula $\chi(x_{g(1)} \cdots x_{g(\ell)})$.
We claim that $\chi^g \in \AFv{k}$. Indeed, any
atom $\alpha$ appearing in $\chi$ is of the form $p(\bx^f_k)$, 
where $p$ is a predicate of some 
arity $m$ and $f \in \bA^m_\ell$. But then the corresponding atom in $\chi^g$ has the form
$\beta\coloneqq \alpha(x_{g(1)} \cdots x_{g(\ell)}) = p(x_{g(f(1))} \cdots x_{g(f(m))}) = p(\bx^{(g \circ f)}_k)$.
Since the composition of adjacent functions is adjacent, the claim follows. 
The following (almost trivial) lemma is useful when manipulating adjacent formulas. 
Recall in this regard that any function $g \in \bA^\ell_k$ 
maps a $k$-tuple $\bar{a}$ over some set to an $\ell$-tuple $\bar{a}^g$ 
over the same set.
% It makes use of our earlier observation 
% that any function $g \in \bA^\ell_k$ maps a $k$-tuple $\bar{a}$
% over some set to an $\ell$-tuple $\bar{a}^g$ over the same set.

\begin{lemma}
  Let $\chi$ be a quantifier-free formula of $\AFv{\ell}$, and  $g \in \bA^\ell_k$. For any $\sig(\chi)$-structure $\fA$ and any $\bar{a} \in A^k$, we have $\fA \models \chi^g[\bar{a}]$ if and only if $\fA \models \chi[\bar{a}^g]$.
  \label{lma:essentiallyTrivial}
\end{lemma}
\begin{proof}
	We may assume without loss of generality that $\chi = p(\bx^f_k)$ is atomic, with $f \in \bA^m_\ell$; the general case follows by an easy structural induction.   
	But then, writing $\bar{a} = a_1 \cdots a_m$,  both sides of the 
	bi-conditional amount to the statement $a_{g(f(1))} \cdots a_{g(f(m))} \in p^\fA$. 
\end{proof}

The adjacent type of any tuple in $\fA$ is thus completely determined 
by that of its primitive generator. 
Indeed, let $\fA$ be a $\sigma$-structure, and 
$\bar{a}$ an $\ell$-tuple from $A$. Then $\bar{a}$ has a primitive generator, say $\bar{b}$ of length $k \leq \ell$, with
$\bar{a} = \bar{b}^g$ for some surjective $g \in \bA^\ell_k$.  Now consider any atomic \AFv{\ell}-formula $\alpha$. 
By Lemma~\ref{lma:essentiallyTrivial} we have 
$\fA \models \alpha[\bar{a}]$ if and only if 
%\fA \models \beta[\bar{b}]$.  
$\fA \models \alpha^g[\bar{b}]$.

When evaluating \AFv{\ell}-formulas, for fixed $\ell$, we can disregard any tuples whose primitive length
is greater than $\ell$. Indeed, consider a pair of $\sigma$-structures
$\fA$ and $\fA'$ over a common~domain~$A$. %We say that
We write
%$\fA$ and $\fA'$ \textit{identical to depth $\ell$}, and write 
$\fA \approx_\ell \fA'$, if, for 
any predicate $p$ (of any arity $m \geq 0$), and any $m$-tuple $\bar{a}$ from $A$ of primitive length at most $\ell$, 
$\bar{a} \in p^\fA$ if and only if $\bar{a} \in p^{\fA'}$. That is, $\fA \approx_\ell \fA'$
just in case
$p^\fA$ and $p^{\fA'}$ agree on all those $m$-tuples whose primitive length~is~at~most~$\ell$. The following may 
be proved by structural induction, using Lemma~\ref{lma:essentiallyTrivial}.
\begin{lemma}
Let $\phi$ be an $\AFv{\ell}$-sentence,  and suppose $\fA$ and $\fA'$
are $\sig(\phi)$-structures over a common domain $A$ such that
$\fA \approx_\ell \fA'$. 
Then $\fA \models \phi \Rightarrow\fA' \models \phi$.
\label{lma:approx}
\end{lemma}
\begin{proof}
  Let $\psi$ be a \textit{formula} of $\AFv{\ell}$ (possibly featuring free variables), 
  and let $k$ ($0 \leq k \leq \ell$) be such that $\psi \in \AFv{[k]}$. (We may as well take the smallest such $k$.)
  We claim that, for any $k$-tuple   of elements $\bar{a}$, $\fA \models \psi[\bar{a}]$ 
  if and only if $\fA' \models \psi[\bar{a}]$.
  To see this, suppose first that $\psi$ is atomic.
  We may write $\psi\coloneqq p(\bx^f_k)$,  where $p$ is a predicate (of arity, say, $m$), and
  $f \in \bA^m_{k}$. If $\bar{a}$ is a $k$-tuple of elements from $A$, then $\fA \models \psi[\bar{a}]$ 
  if and only if $\bar{a}^f \in p^\fA$. But the primitive length of $\bar{a}^f$ is certainly at most $k = |\bar{a}|$. 
  This proves our claim for all $k$ ($0 \leq k \leq \ell$) and for all \textit{atomic} $\psi \in \AFv{[k]}$.
  The general case follows simply by structural induction. The statement of the
  lemma is the special case where $\psi$ has no free variables.
\end{proof}

In view of Lemma~\ref{lma:approx}, when considering models of \AFv{\ell}-sentences, it will be useful to take the extensions of predicates (of whatever arity) to be \textit{undefined in respect of tuples whose primitive length is greater than} $\ell$, since these 
cannot affect the outcome of semantic evaluation. That is, where $\ell$ is clear from context, we typically suppose any
model $\fA$ of $\phi$ to 
determine  whether $\bar{a} \in p^\fA$ for any $m$-ary predicate $p$ and any $m$-tuple $\bar{a}$ \textit{of primitive length at most} $\ell$; but with respect to $m$-tuples $\bar{a}$ having greater primitive length, $\fA$ remains agnostic. To make it clear that the structure
$\fA$ need not be fully defined, we speak of  such a structure~$\fA$ as a {\em layered} structure, and we refer to $\ell$ 
as its {\em primitive length}. Notice that the notion of primitive length is independent of the arities of the predicates interpreted. 
A layered structure $\fA$ may have primitive length, say 3, but still interpret a predicate $p$ of arity, say, 5. In this case, it is
determined whether $\fA \models p[babcc]$, because the primitive generator of $babcc$ is $abc$; however, it will not be determined
whether  $\fA \models p[abcab]$, because $abcab$~is~primitive.

One of the intriguing aspects of layered structures is that they allow us to build up models of \AF-formulas layer by layer.
Suppose $\fA$ has primitive length $k$; and we wish to
construct a layered structure $\fA^+$ of primitive length $k{+}1$ over the same domain $A$, agreeing with the assignments made by $\fA$. Clearly, it suffices to fix the adjacent type of each primitive $(k{+}1)$-tuple $\bar{b}$ from $A$. To fix the adjacent type
of $\bar{b}$---and hence that of its reversal, $\bar{b}^{-1}$---we consider each predicate $p$ in turn---of arity, say, $m$---and decide, for any
$m$-tuple $\bar{c}$ from $A$ whose primitive generator is $\bar{b}$, whether $\fA \models p[\bar{c}]$. 
%Once this has been done, $\atp^\fA[\bar{b}]$ and $\atp^\fA[\bar{b}^{-1}]$ will be fixed. 
Now repeat this process for all pairs of mutually inverse 
primitive words $(\bar{b}, \bar{b}^{-1})$ from $A$ having primitive length~$k{+}1$.
Since every tuple $\bar{c}$ considered for inclusion in the extension of some predicate has primitive length~$k{+}1$, these assignments will not clash with any previously made
in the original structure~$\fA$. Moreover, since, by Theorem~\ref{theo:main}, every $m$-tuple $\bar{c}$ assigned in this process has a unique primitive generator
$\bar{b}$ (up to reversal), these assignments will not clash with each other.
Thus, to increment the primitive length of
$\fA$, one takes each inverse pair $(\bar{b}, \bar{b}^{-1})$ of primitive $(k+1)$-tuples in turn,
and fixes the adjacent type of each $\bar{b}$ consistently with the existing assignments of all tuples generated by
proper infixes of $\bar{b}$, as given in the original structure $\fA$.

We finish this section with an easy technical  observation  that will be needed in the sequel.
Denote by  $\vec{\bA}^m_k$ the set of all functions $f \in \bA^m_k$ such that $f(m) = k$.
Thus, if $f \in \vec{\bA}^m_k$ is used to define a walk of length $m$ on some word $\bar{a}$ of length $k$, then the walk in question ends at the final position of $\bar{a}$. 
\begin{lemma}
	Let $\bar{c}$ be a word of length $m \geq 0$ over some alphabet $A$, and $d$ an element of~$A$ that does not appear in $\bar{c}$. If $\bar{c}d$ is not primitive, then neither is $\bar{c}$. In fact, there is a word~$\bar{a}$
	of length $k < m$ and a function $f \in \vec{\bA}^m_{k}$ such that $\bar{a}^f = \bar{c}$.
	\label{lma:simplePrimitive}
\end{lemma}
\begin{proof}
  Suppose $\bar{c}d = \bar{b}^g$ for some word $\bar{b}$ of length $k+1 \leq m$ and some surjective $g \in \bA^{m+1}_{k +1}$.
  Since $d$ does not occur in $\bar{c}$, it is immediate that $d$ occupies either the first or last position in $\bar{b}$, for
  otherwise, it would be encountered again in the entire traversal of $\bar{b}$ (as $g$ is adjacent and surjective). By reversing $\bar{b}$ if necessary,
  assume the latter, so that we may write $\bar{b} = \bar{a}d$, with $g(m+1) = k+1$. By adjacency, $g(m) = k$, so that
  setting $f = g \setminus \set{\langle m+1, k+1\rangle}$, we have the required $\bar{a}$ and $f$.
\end{proof}

Finally, we remark that, if $f \in \vec{\bA}^m_{k}$, then the function $f^+ = f \cup \langle m+1,k+1 \rangle$
satisfies $f^+ \in \vec{\bA}^{m+1}_{k+1}$. That is, we can extend $f$ by setting $f(m+1) = k+1$, retaining adjacency.

%% file: sections/upper_bound_AF.tex
%!TEX root = ../main.tex
\section{Upper bounds for $\AF{}$ and $\AFv{k}$}\label{section:AF-upper-bounds}
In this section, we establish a small model property for each of the fragments $\AFv{k}$ with $k \geq 3$.
Define the function $\ft(k,n)$ inductively by $\ft(0,n) = n$ and $\ft(k{+}1,n) = 2^{\ft(k,n)}$.
We show that, for some fixed polynomial $p$, if $\phi$ is a satisfiable formula
of  $\AFv{k}$, then $\phi$ is satisfied in a structure of size at most $\ft(k-2,p(\sizeof{\phi}))$. We proceed by induction, establishing first the base case for $k = 3$, and then reducing the case $k{+}1$ to the case~$k$. It follows that
the satisfiability problem (= finite satisfiability problem) for $\AFv{k}${} is in $(k{-}2)$-\NExpTime{} for all $k \geq 3$. 
The best lower complexity bound is $\lfloor k/2 \rfloor$-\NExpTime-hard, from the $k$-variable fluted fragment~\cite{phst19}.
For $k \leq 2$, the adjacency restriction has no effect on the complexity of satisfiability. Thus satisfiability for $\AFv{2}$ is \NExpTime-complete, while for $\AFv{1}$ and $\AFv{0}$ it is \NP-complete.
We begin by establishing a normal form lemma for \AF.
\begin{lemma}
Let $\phi$ be a sentence of \AFv{\ell+1}, where $\ell \geq 2$. We can compute, in polynomial time, an \AFv{\ell+1}-formula 
$\psi$ satisfiable over the same domains as $\phi$, of the form
\begin{equation}
\bigwedge_{i \in I} \forall \bx_{\ell} \exists x_{\ell+1}\, \gamma_i \wedge 
\forall \bx_{\ell+1}\,
\delta,
\label{eq:anf}
\end{equation}  
where $I$ is a finite index set, 
the formulas $\gamma_i$ and $\delta$ are quantifier-free.
\label{lma:anf}
\end{lemma}

% To simplify the presentation, we shall always assume that any normal-form formula $\phi$ as given by~\eqref{eq:anf} over a signature
% $\Sigma$ satisfies
% $|\Sigma| > |I|$. There is no loss of generality involved here, since we may replace each conjunct $\forall \bx_{\ell} \exists x_{\ell+1}\, \gamma_i$ ($i \in I$) by  
% $\forall \bx_{\ell} \exists x_{\ell+1}\, (\gamma_i \wedge p_i(\bx_\ell x_{\ell{+}1}))$ with $p_i$ a fresh $(\ell{+}1)$-ary predicate,
% without
% affecting the domains over which $\phi$ is satisfiable. Further, we may as well assume that $\Sigma$ contains no 
% predicates of arity 0 (proposition letters), since their truth-values may be guessed.

Let $\phi$ be a normal-form \AFv{\ell+1}-formula as given in~\eqref{eq:anf}, over signature $\sigma$.
Recall the operation ${\cdot}^f$ on quantifier-free adjacent formulas employed in Lemma~\ref{lma:essentiallyTrivial}, as well as
the sets of functions $\vec{\bA}^\ell_k$ employed in Lemma~\ref{lma:simplePrimitive}. For
any $f \in \vec{\bA}^\ell_{k}$, we continue to write 
$f^+$ for the function (in $\vec{\bA}^{\ell+1}_{k+1}$) extending $f$ by setting $f(\ell+1) = k+1$.
Now define the {\em adjacent closure of $\phi$}, denoted $\phi^\#$, to be:
\begin{equation*}
\bigwedge_{\hspace{2mm}i \in I\phantom{\bigwedge^{\ell+1}_{k+1}}\hspace{-5mm}} \hspace{-3mm}
\hspace{-1mm} \bigwedge_{\hspace{2mm}k=1\phantom{\bigwedge^{\ell+1}_{k+1}}\hspace{-4mm}}^{\ell-1} \hspace{-3mm}
\bigwedge_{f \in \vec{\bA}^{\ell}_{k}} \hspace{-2mm}
\forall \bx_{k} \exists x_{k+1}\, \gamma^{f^+}_i 
\wedge
\bigwedge_{\hspace{2mm}k=1\phantom{\bigwedge^{\ell+1}_{k+1}}\hspace{-5mm}}^{\ell}  \hspace{-4mm}
\bigwedge_{\hspace{2mm}g \in \bA^{\ell+1}_{k}\phantom{\bigwedge^{\ell+1}_{k+1}}\hspace{-8mm}} \hspace{-3mm}
\forall \bx_{k}\, \delta^g.
\end{equation*}  
Observe that the
conjunctions for the $\forall^\ell\exists$-formulas range over $\vec{\bA}^{\ell}_{k}$, while the conjunctions
for the purely universal formula range over the whole of $\bA^{\ell+1}_{k}$.
Up to trivial logical rearrangement and re-indexing of variables, $\phi^\#$ is actually a normal-form \AFv{\ell}-formula. 
In effect, $\phi^\#$ is the result of identifying various universally quantified variables in $\phi$ in a way which preserves adjacency.
The following~lemma~is~therefore~immediate. 
\begin{lemma}
  Let $\phi \in \AFv{\ell+1}$ be in normal-form. Then $\phi \models  \phi^\#$. 
  \label{lma:adjacentClosure}
\end{lemma}
The following notation will be useful.
If $\chi$ is any quantifier-free $\AFv{\ell+1}$-formula, we denote by $\chi^{-1}$ the formula $\chi(x_{\ell+1}, \dots, x_1)$
obtained by simultaneously replacing each variable $x_h$ by $x_{\ell -h +2}$ ($1 \leq h \leq \ell+1$); and we denote by $\hat{\chi}$ the formula $\chi \wedge \chi^{-1}$. Obviously $\chi^{-1}$ and $\hat{\chi}$ are also in $\AFv{\ell+1}$. 
If $\eta$ is an adjacent $\ell$-type, we denote by $\eta^+$ the 
quantifier-free $\AFv{\ell+1}$-formula $\eta(x_2, \dots, x_{\ell{+}1})$ obtained by incrementing the index of each variable. 
Finally, if $\chi$ is a quantifier-free $\AFv{\ell+1}$-formula over some signature $\sigma$ (which we take to be given by context), we
denote by $\chi^\circ$ the quantifier-free $\AFv{\ell}$-formula
$\smash{\bigvee \set{\eta \in \Atp^\sigma_\ell \mid \text{$\chi \wedge \eta^+$ is consistent}}}$. The intuition in this last case is that, if $a$ is an element  and $\bar{a}$ an $\ell$-tuple of elements such that $a\bar{a}$ satisfies $\chi$ in some structure, then
$\chi^\circ$ is the strongest statement that follows regarding $\bar{a}$. 

Now we are in a position to tackle the main task of this section, namely, to bound the complexity of the satisfiability 
problem for $\AFv{k}$ ($k \geq 3$). Certainly the satisfiability problem for $\AFv{2}$ is in \NExpTime, since any normal-form $\AFv{2}$-formula is in $\FOt$.
Here, we strengthen that result to $\AFv{3}$ (which will sharpen the bound of Theorem~\ref{theo:upperBound} by one exponential). The proof is 
similar to an analogous result for the three-variable \textit{fluted} fragment, $\FLv{3}$~\cite[Lemma~4.5]{phst19}

Let $\sigma$ be a relational signature. 
If $\pi$ is a 1-type over $\sigma$, define the 2-type $\pi^2$, over the same
signature, to be $\set{\lambda \mid \text{$\lambda$ a literal in \AFv{[2]} s.t.~$\lambda(x_1, x_1) \in  \pi$}}$. 
The intuition here is that
if $\pi$ is the type of an element $a$ in some structure, then $\pi^2$ is the type of the pair $aa$.
A {\em connector-type} 
(\textit{over} $\sigma$) is a set $\omega$ of 2-types over $\sigma$ subject to the condition that there exists some 1-type $\pi$ over $\sigma$ such that $\pi^2 \in \omega$ and $ \zeta \models \pi$ for all $\zeta \in \omega$. 
This 1-type $\pi$ is clearly unique, and we denote it by $\tp(\omega)$. If $\fA$ is any structure interpreting $\sigma$ and $a \in A$, then $a$ defines a connector-type $\omega$ over $\sigma$ in a natural way by
setting $\omega = \set{\tp^\fA[a,b] \mid b \in A}$. We refer to $\omega$ as {\em the connector-type of $a$ in $\fA$}, and  denote it $\con^\fA[a]$. It follows immediately from the above definitions that $\tp(\con^\fA[a]) = \tp^\fA[a]$.
When speaking of connector-types, we suppress
reference to $\sigma$ if irrelevant or clear from context. 

Let $\phi$ be a normal-form formula of $\AFv{3}$, as given in~\eqref{eq:anf}, with $\ell = 2$ and $\sigma = \sig(\phi)$.
In the sequel we refer freely to the subformulas $\gamma_i$ ($i \in I$) and $\delta$ of $\phi$. Say that a connector-type~$\omega$ is {\em compatible} with $\phi$ if the following conditions hold:
\begin{description}
	\item[$\mathrm{L}\exists_1$:] for all $i \in I$,  
	there exists $\eta \in \omega$~s.t.~$\eta \models \gamma_i(x_1,x_1,x_2)$. 
	\item[$\mathrm{L}\exists_2$:] for all $\zeta$ such that $\zeta^{-1} \in \omega$ and all $i \in I$,  
	there exists $\eta \in \omega$ 
	such that the \AFv{3}-formula $\zeta \wedge \eta^+ \wedge \gamma_i \wedge \hat{\delta}$~is~consistent;
	\item[$\mathrm{L}\forall_1$:] for all $\eta \in \omega$ and all $f \in \bA^3_2$, $\eta \models \delta^f$;
	\item[$\mathrm{L}\forall_2$:] for all $\zeta$ such that $\zeta^{-1} \in \omega$ and all $\eta \in \omega$,  
	the \AFv{3}-formula $\zeta \wedge \eta^+ \wedge \hat{\delta}$ is consistent.
\end{description}

The proofs of Lemmas~\ref{lma:simpleConnector}--\ref{lma:smallConnector} are straightforward and will be omitted.
\begin{lemma}
	If $\phi$ is a normal-form $\AFv{3}$-formula, $\fA \models \phi$ and $a \in A$, then  
	$\con^\fA[a]$ is compatible with $\phi$.
	\label{lma:simpleConnector}
\end{lemma}

A set $\Omega$ of connector-types is said to be {\em coherent} if 
the following conditions hold:
\begin{description}
	\item[$\mathrm{G}\exists$:] for all $\omega \in \Omega$ and all $\zeta \in \omega$, there exists $\omega' \in \Omega$ 
	such that $\zeta^{-1} \in \omega'$;
	\item[$\mathrm{G}\forall$:] for all $\omega, \omega' \in \Omega$,  
	there exists a 2-type $\zeta$ such that $\zeta \in \omega$ and $\zeta^{-1} \in \omega'$.
\end{description}
\begin{lemma}
	Let $\fA$ be a structure. Then $\Omega = \set{\con^\fA[a] \mid a \in A}$ is coherent.
	\label{lma:simpleConnectorSet}
\end{lemma}
Define a {\em certificate} for $\phi$ to be a non-empty, coherent set of connector-types, all of which are compatible with $\phi$.
\begin{lemma}
	Any satisfiable normal-form $\AFv{3}$-formula has a certificate $\Omega$ such that both~$|\Omega|$ and $|\bigcup \Omega|$ are  $2^{O(\sizeof{\phi})}$.
	\label{lma:smallConnector}
\end{lemma}

We are now in a position to obtain a bound on the size of models of $\AFv{3}$-formulas. 
\begin{lemma}
	Let $\phi$ be a normal-form $\AFv{3}$-formula over a signature $\sigma$.
	If $\phi$ is satisfiable, then it has a model of size~$2^{O(\sizeof{\phi})}$.
	\label{lma:base}
\end{lemma}
\begin{proof}
	We may assume without loss of generality that $\sigma$ features no proposition letters.
	Let~$\phi$ be as given by~\eqref{eq:anf}.
	By Lemma~\ref{lma:smallConnector}, $\phi$ has a certificate $\Omega$ of cardinality at most $2^{O(\sizeof{\phi})}$; moreover the 
	set of 2-types $T$ occurring anywhere in $\Omega$ is $2^{O(\sizeof{\phi})}$.
	Let $H= \set{0, 1, 2}$, let $I$ be the index set occurring in $\phi$, let
	$J$ be a set of cardinality $343 = 7^3$, and let $g\colon J^2 \rightarrow J$ be a function satisfying the
	conditions of Lemma~\ref{lma:simpleComb} with $k = 2$. Defining $A = \Omega \times T \times H \times I \times J$, 
	we see that $|A|$ is $2^{O(\sizeof{\phi})}$, as required by the lemma. We write any element $a \in A$ as $(\omega,\zeta,h,i,j)$. 
	We shall construct a layered model $\fA \models \phi$ of primitive length 3 over this domain, proceeding  layer by layer. 
	In the sequel, bear in mind that a pair or triple of elements is primitive if and only if those elements are distinct.
	
	\medskip
	
	\noindent
	{\bf Stage 1: }
	We set the 1-type of any $a = (\omega,\zeta,h,i,j)$ to be $\tp^\fA[a] = \tp(\omega)$. Clearly, all these
	determinations can be made independently, since $\sigma$ features no proposition letters. At this point, we have a layered structure of primitive length 1.
	
	\medskip
	
	\noindent
	{\bf Stage 2: } Now consider any $a = (\omega,\zeta,h,i,j) \in A$ and
	any $\eta \in \omega$. By (G$\exists$), there exists~$\omega_\eta \in \Omega$
	such that $\eta^{-1} \in \omega_\eta$. For each $i' \in I$ and $j' \in J$ set $\tp^\fA[a,a_{i',j'}] = \eta$, where $a_{i',j'}$ denotes
	the element $(\omega_\eta, \eta, h{+}1,i',j')$. (Here the addition in ``$h{+}1$'' is taken modulo 3.) The index~$\eta$ ensures that the $a_{i',j'}$ are chosen to be distinct for 
	distinct $\eta \in \omega$. Moreover, the index $h{+}1$ ensures that this process can be carried out for every $a \in A$ without danger of clashes. (This is the familiar technique of `circular witnessing'~\cite{GradelKV97}.) Finally, suppose $a = (\omega,\zeta,h,i,j)$ and $a' = (\omega',\zeta',h',i',j')$ are distinct elements of $A$ for which
	$\tp^\fA[a,a']$ has not yet been defined. By (G$\forall$), there exists $\eta \in \omega$
	such that $\eta^{-1} \in \omega'$, and we set $\tp^\fA[a,a'] = \eta$. At the end of this process,
	all 1- and 2-types have been defined, and we thus have a layered structure of primitive length 2.
	From the foregoing construction, if $a = (\omega,\zeta,h,i,j) \in A$ and $\eta \in \omega$, then  there exists
	a constructor-type $\omega'$ such that $\tp^\fA[a,b] = \eta$
	for each $b \in A$ of the form $(\omega',\eta,h+1,i',j')$ (where $i' \in I$, $j' \in J$); moreover,
	for all $a = (\omega,\zeta,h,i,j)$ and $b = (\omega',\zeta',h',i',j')$ with
	$\tp^\fA[a,b] = \eta$, we are guaranteed that $\eta \in \omega$ and $\eta^{-1} \in \omega'$.
	We remark that, in particular,
	$\con^\fA[a] = \omega$. It follows from L$\exists_1$
	that, for every $a \in A$ and every $i \in I$, there exists $b \in A$ such that $\gamma_i[a,a,b]$. 
	Another way of saying this is that, for every pair of elements $a_1,a_2$ whose primitive length is 1 (i.e.~$a_1 = a_2$), 
	$\fA$ provides a witness for the formula $\exists x_3\, \gamma_i$. Likewise, it follows from
	L$\forall_1$
	that,
	for every triple $\bar{a}$ whose primitive length is either 1 or 2, $\fA \models \delta[\bar{a}]$. Indeed, if
	$\bar{a} = \bar{b}^f$ where $|\bar{b}| \leq 2$, we have $\fA \models \delta^f[\bar{b}]$, whence $\fA \models \delta[\bar{a}]$ by Lemma~\ref{lma:essentiallyTrivial}. 
	
	\medskip
	
	\noindent
	{\bf Stage 3: }
	We now increment the primitive length of $\fA$~to~$3$ by setting the adjacent 3-types of all primitive triples in~$\fA$.
	Fix any pair of distinct elements $a = (\omega, \tau, h, i, j)$ and $a' = (\omega', \tau', h', i', j')$. Let us write
	$\zeta= \tp^\fA[a,a']$, so that,
	by construction of $\fA$ in the previous stage, $\zeta \in \omega$ and $\zeta^{-1} \in \omega'$. By (L$\exists_2$), there exists 
	some $\eta \in \omega'$ such that the \AFv{3}-formula $\psi\coloneqq \zeta \wedge \eta^+ \wedge \gamma_i \wedge \hat{\delta}$ is consistent; 
	let $\theta_i$ be an adjacent 3-type entailing this formula.
	By the construction of the previous stage again, we can find an element $b_i\coloneqq (\omega'', \eta, h'+1, i, g(j,j')) \in A$ such that 
	$\tp^\fA[a',b_i] = \eta$. We shall set $\atp^\fA[a,a',b_i] = \theta_i$ for all $i \in I$. From the index $i$, the elements $b_i$ are distinct,
	and so these assignments do not clash with each other. Since $\theta_i$ entails $\zeta \wedge \eta^+$, they do not clash with the 2-types
	assigned so far. Since $\theta_i$ entails $\gamma_i$, the pair $a, a'$ now has a witness in respect of the formula $\exists x_3\, \gamma_i$. 
	From property (i) of $g$ secured by Lemma~\ref{lma:simpleComb}, the triple $a,a',b_i$ is primitive; hence the only primitive triples whose
	adjacent types are thereby defined are $a,a',b_i$ and $b_i,a',a$. But since $\theta_i$ entails $\hat{\delta}$, neither of these triples violates 
	$\forall x_1 x_2 x_3\, \delta$. Now repeat this construction for all pairs of  distinct elements $a = (\omega, \tau, h, i, j)$ and $a' = (\omega', \tau', h', i', j')$. We claim that no tuple~$\bar{c}$ is assigned to the extensions of any predicates twice in this process. Since
	$\bar{c}$ must have some primitive generators $a_1a_2a_3$ and $a_3a_2a_1$, the only possibility for double assignment of
	$\bar{c}$ is if~$a_3$ is chosen as some witness for the pair $a_1, a_2$, and $a_1$ is chosen as some witness for the pair $a_3, a_2$.
	Remembering that $a_1$, $a_2$ and $a_3$ are actually quintuples,
	let their final components be, respectively, $j$, $j'$, $j''$. By the choice of witnesses, $j'' = g(j, j')$ and $j= g(j'',j')$. But this contradicts property (ii) of $g$ secured by Lemma~\ref{lma:simpleComb}, thus establishing the claim that no primitive triple
	is assigned to extensions of predicates twice.
	At this point, for
	every pair of elements $a_1a_2$ (of primitive length either 1 or 2) and every $i \in I$, 
	$\fA$ provides a witness for the formula $\exists x_3\, \gamma_i$. Moreover, no adjacent 3-type so-far assigned violates
	$\delta$. 
	To complete the extension of $\fA$ to primitive length 3, it remains only to assign 
	adjacent types to all remaining primitive triples without violating $\delta$. Suppose, then $a,a',a''$ are distinct elements whose adjacent type in $\fA$ has not yet been defined. Let $\zeta = \tp^\fA[a_1,a_2]$ and $\eta = \tp^\fA[a_2,a_3]$. By the previous stage, 
	$\zeta \wedge \eta^{+} \wedge \hat{\delta}$ is consistent, so let $\theta$ be an adjacent 3-type entailing this formula, and set 
	$\tp^\fA[a_1,a_2,a_3] = \theta$. Observe that we are also thereby assigning the adjacent 3-type of $\tp^\fA[a_3,a_2,a_1]$, but are assigning
	no other adjacent 3-types. Since $\theta$ entails $\zeta \wedge \eta^+$, this assignment does not clash with the assignments of the previous step. Since $\theta$ entails $\hat{\delta}$, no newly assigned triple violates $\delta$. This completes the construction of the model~$\fA$.
\end{proof}

Extending Lemma~\ref{lma:base} to the whole of $\AF$ represents a greater challenge. For the next two lemmas (\ref{lma:reductionDirection1} and \ref{lma:reductionDirection2}), fix a normal-form \AFv{\ell{+}1}-formula $\phi$ over some signature $\sigma$, as given in~\eqref{eq:anf},
with $\ell \geq 3$.
We construct a normal-form formula $\phi' \in \AFv{\ell}$ such that: (i) if $\phi$ is satisfiable over some domain $A$, then so is $\phi'$; and (ii) if $\phi'$ is satisfiable
over some domain $B$, then $\phi$ is satisfiable over a domain  $A$, with $|A|/|B|$ 
bounded by some exponential~function~of~$\sizeof{\phi}$. 

Recall that the adjacent closure, $\phi^\#$ of $\phi$, may be regarded as a normal-form \AFv{\ell}-formula over the same signature.
For every adjacent $\ell$-type $\zeta$ over $\sigma$,
let $p_\zeta$ be a fresh predicate of arity $\ell{-}1$. Intuitively, we shall think of $p(x_2 \cdots x_{\ell})$ as stating ``{\em for some $x_1$},
the $\ell$-tuple $\bx_\ell= x_1 \cdots x_{\ell}$ is of adjacent type $\zeta$''. Now define $\phi'$ to be the conjunction of $\phi^\#$ with the following \AFv{\ell}-formulas:
\begin{align}
& \bigwedge_{\zeta \in \Atp^\sigma_\ell} \forall \bx_{\ell} \big(\zeta \rightarrow p_\zeta(x_2 \cdots x_{\ell})\big)
\label{eq:reduction1}\\
& \bigwedge_{\zeta \in \Atp^\sigma_\ell} \hspace{-2mm} \bigwedge_{\hspace{3mm} i \in I\phantom{\Atp^\sigma_\ell} \hspace{-4mm}} \hspace{-2mm} \forall \bx_{\ell-1} \exists x_{\ell} \big(p_\zeta(\bx_{\ell-1}) \rightarrow
(\zeta \wedge \hat{\delta} \wedge \gamma_i)^\circ \big)
\label{eq:reduction2}\\
& \bigwedge_{\zeta \in \Atp^\sigma_\ell} \forall \bx_{\ell} \big(p_\zeta(\bx_{\ell-1}) \rightarrow 
(\zeta \wedge \hat{\delta})^\circ\big)
\label{eq:reduction3}
\end{align}
\begin{lemma}
Suppose $\fA \models \phi$. Then we can expand $\fA$ to a model $\fA^+ \models \phi'$.
\label{lma:reductionDirection1}
\end{lemma}
\begin{proof}
Set
$p^{\fA^+}_\zeta = \set{\bar{a} \in A^{\ell{-}1} \colon \text{$\fA \models \zeta[a\bar{a}]$ for some $a \in A$}}$,
for every $\zeta \in \Atp^\sigma_\ell$. 
The truth of~\eqref{eq:reduction1} in $\fA^+$ is then immediate. To see the same for~\eqref{eq:reduction2}, fix 
$\zeta \in \Atp^\sigma_\ell$ and $i \in I$, and suppose 
$\fA^+ \models p_\zeta(\bar{a})$, where $\bar{a} \in A^{\ell-1}$. Then there exists $a \in A$ such that 
$\fA \models \zeta[a\bar{a}]$. Moreover, since $\fA \models \phi$, there exists $b \in A$ such that $\fA \models \gamma_i[a\bar{a}b]$ and 
$\fA \models \hat{\delta}[a\bar{a}b]$. Now let $\eta= \atp^\fA[\bar{a}b]$. Writing $\chi$ for
$\zeta \wedge \hat{\delta} \wedge 
\gamma_i$, we have $\fA \models \chi[a\bar{a}b]$,
whence $\chi$
is consistent; and since $\fA \models \eta[\bar{a}b]$, it follows that
$\eta \models \chi^\circ$. Thus, $b$ is a witness for the $(\ell{-}1)$-tuple $\bar{a}$ required by
the relevant conjunct of~\eqref{eq:reduction2}. This secures the truth of~\eqref{eq:reduction2} in $\fA^+$.
Formula~\eqref{eq:reduction3} is handled similarly.
\end{proof}
\begin{lemma}
  Suppose $\fB \models \phi'$. Then we can construct a model $\fC^+ \models \phi$ such that $|C^+|/|B| \leq |I| \cdot (\ell^2+\ell+1)^\ell$. 
  \label{lma:reductionDirection2}
\end{lemma}
\begin{proof}

Since $\phi' \in \AFv{\ell}$, we may assume  by Lemma~\ref{lma:approx} that $\fB$ is a layered structure of primitive length $\ell$---that is, does not specify the extensions of predicates in respect of tuples whose primitive length is greater than $\ell$. 
Let $\fB^-$ be the reduct of $\fB$ to the signature $\sigma$ (i.e.~we forget the predicates $p_\zeta$).
Thus, every $\ell$-tuple from $B$ satisfies a unique element of $\Atp^\sigma_\ell$.
We first define a collection of `witness' functions $v_{i} \colon B^{\ell}\rightarrow B$, where $i \in I$. 
For any $\ell$-tuple $\bar{b} = b_1 \cdots b_\ell$, let $\zeta = \atp^{\fB^-}[\bar{b}]$.
By~\eqref{eq:reduction1}, $\fB \models p_\zeta[b_2 \cdots b_\ell]$, whence, by~\eqref{eq:reduction2}, we may select $b \in B$
such that $\fB \models (\zeta \wedge \hat{\delta} \wedge \gamma_i)^\circ[(b_2 \cdots b_k b)]$. Set 
$v_{i}(\bar{b}) = b$.
Now let $J$ be a set of cardinality $\ell^2 + \ell + 1$ and let $g \colon J^\ell \rightarrow J$ a function satisfying conditions 
(i) and (ii) guaranteed by
Lemma~\ref{lma:simpleComb}. 
We inflate the structure $\fB^-$ using the product construction of Lemma~\ref{lma:cartesian}.
Specifically,  we define  $\fC=\fB^- \times (I \times J)$, writing elements of $\fC$ as triples $(b,i,j)$, where $b \in B$, $i \in I$ and
$j \in J$. Now, predicate extensions featuring tuples of primitive length greater than $\ell$ can be safely disregarded in the structure $\fC$.
We next define a collection of witness
functions $w_{i} \colon C^{\ell}\rightarrow C$, based on the functions $v_{i}$
defined above. The motivation is that these functions will allow us to choose witnesses in $\fC$ for the
conjuncts~\eqref{eq:reduction2} that do not, as it were, tread on each others' toes. 
Consider any $\ell$-tuple  $\bar{c} = c_1 \cdots c_\ell$ of elements in $C$, with $c_h= (b_h,i_h,j_h)$ for each $h$ ($1 \leq h \leq \ell$). 
Writing $\bar{b} = b_1 \cdots b_\ell$, we define $w_{i}(\bar{c})$ to be the element $(v_{i}(\bar{b}), i, g(j_1 \cdots j_\ell))$. 
Since $\fB^- \models (\zeta \wedge \hat{\delta} \wedge \gamma_i)^\circ[b_2 \cdots b_\ell\, w_{i}(\bar{b})]$, 
it follows from Lemma~\ref{lma:cartesian} that $\fC \models (\zeta \wedge \hat{\delta} \wedge \gamma_i)^\circ[c_2 \cdots c_\ell\, w_{i}(\bar{c})]$. 
In addition, the functions $w_i$ satisfy the following two additional properties:
\begin{description}
  \item[(w1)] for fixed $\bar{c}$, the $w_{i}(\bar{c})$
  are distinct as $i$ varies over $I$;
  \item[(w2)] $w_{i}(\bar{c})$ does not occur in $\bar{c}$; 
  \item[(w3)] 
  if $\bar{c}'$
  is an $\ell$-tuple consisting of the elements
  $c_2, \dots, c_\ell, w_{i}(\bar{c})$ in some order, then 
  $w_{i'}(\bar{c}')$ does not occur in $\bar{c}$ for any $i' \in I$.
\end{description}
Indeed, {(w1)} is immediate from the fact  $w_i(\bar{c})$ contains $i$ as its second element; {(w2)} and {(w3)}
follow, respectively, from conditions (i) and 
(ii) on $g$ guaranteed in
Lemma~\ref{lma:simpleComb}. 

We are now ready to extend $\fC$ to a structure $\fC^+$ of primitive length $\ell{+}1$ such that $\fC^+ \models \phi$.  
We first manufacture witnesses required by the conjuncts 
$\forall \bar{x} \exists x_{\ell+1} \gamma_i$, insofar as these are not already present.
Fix any $\ell$-tuple $\bar{c} = c_1 \cdots c_\ell$,
and let $\zeta= \atp^{\fC}[\bar{c}]$. 
Now consider any $i \in I$, and write $c = w_{i}(\bar{c})$. 
We have two cases, depending on whether the word 
$\bar{c}c$ is primitive. Suppose first that it is not. By {(w2)}, $c$ is not an element of $\bar{c}$, whence
by Lemma~\ref{lma:simplePrimitive}
there is some $k$-tuple $\bar{d}$ ($k < \ell$) and $f \in \vec{\bA}^{\ell}_{k}$ such that $\bar{d} = \bar{c}^f$. 
As before, define $f^+ \in \vec{\bA}^{\ell+1}_{k+1}$ extending $f$ by setting $f(k+1) = \ell+1$.
Since $k < \ell$, and $\fC \models
{\phi}^\#$, there exists $c' \in C$ such that $\fC \models \gamma^{f^+}_i[\bar{d}c']$.
By Lemma~\ref{lma:essentiallyTrivial}, $\fC \models \gamma_i[(\bar{d}c')^{f^+}]$, or in other words, $\fC \models \gamma_i[\bar{c}c']$, so that a witness $c'$ is already present in respect of the tuple $\bar{c}$ and
the index $i$. (Notice that we are throwing our original witness, $c$, away.)
Suppose on the other hand that $\bar{c}c$ is primitive. Since $\fC$ has primitive length~$\ell$, no tuple with primitive generator $\bar{c}c$
has been assigned to the extension of any predicate in $\fC$.
Let $\eta = \atp^\fC[c_2 \cdots c_\ell c]$. Writing $\chi$ for the $\AFv{\ell}$-formula $\zeta \wedge \hat{\delta} \wedge \gamma_i$ 
it follows from the choice of $c$ that $\fC \models \chi^\circ[\bar{c}c]$, whence,  by the definition of the operator $(\cdot)^\circ$, the 
$\AFv{\ell+1}$-formula $\eta^+ \wedge (\zeta \wedge \hat{\delta} \wedge \gamma_i)$ is consistent.
Therefore, there exists an adjacent $(\ell+1)$-type $\omega$
entailing it, and we may fix $\atp^{\fC^+}[\bar{c}c] = \omega$. To see that this assignment makes sense and extends $\fC$, recall that
$\atp^{\fC^+}[\bar{c}c]$ specifies whether $\fC^+ \models q[\bar{d}]$ for any $m$-tuple~$\bar{d}$ whose primitive
generator is an infix, say $\bar{e}$, of $\bar{c}c$. If $\bar{e}$ is of length $\ell$ or less, then its adjacent type
has already been fixed in $\fC$ consistently with $\zeta$ or $\eta^+$. Otherwise, the primitive generator of $\bar{d}$ is
$\bar{c}c$, so that $\fC$ does not determine satisfaction of $q$ by $\bar{d}$; writing $\bar{d} = (\bar{c}c)^f$, then,
we may set $\fC^+ \models q[\bar{d}]$ if and only if $\models \omega \rightarrow q((\bx_\ell x_{\ell+1})^f)$.
Since $\omega \models \gamma_i$, we see that,
following these assignments, $\fC^+$ has been provided with a witness in respect of the tuple $\bar{c}$ and
the index $i$. We claim in addition that the newly assigned tuples do not violate $\forall \bx_\ell x_{\ell+1}\, \delta$. For suppose
that $\bar{d}$ is an $(\ell+1)$-tuple whose adjacent type in $\fC^+$ has been defined. If the primitive length of $\bar{d}$ is $\ell$ or less, then we have
$\bar{d} = \bar{e}^g$ for some primitive $\bar{e}$ of length $k \leq \ell$ and some $g \in \bA^{\ell+1}_{k}$. Since $\fC \models \phi^{\#}$, we have $\fC \models \delta^g[\bar{e}]$, whence by Lemma~\ref{lma:essentiallyTrivial}, 
$\fC \models \delta[\bar{d}]$. If, however, 
the primitive length of $\bar{d}$ is $\ell+1$, then  $\bar{d}$ is either $\bar{c}c$ or its reversal, and by the fact that 
$\omega \models \hat{\delta}$, we have $\fC^+ \models \delta[\bar{d}]$ as required. Still keeping
$\bar{c}$ fixed for the moment, we may carry out the above procedure for all $i \in I$. 
To see that these assignments do not interfere with each other, we simply note property {(w1)} of the functions~$w_i$.% established above.

Now make these assignments as just described for \textit{each} word $\bar{c} \in C^\ell$. To ensure that these assignments do not interfere with each other,
we make use of properties {(w1)} and {(w3)} of the functions $w_i$. If $\bar{d}$ is an $m$-tuple that has been assigned (or not) to the extensions of various predicates by the process described above, then the two primitive generators of $\bar{d}$ must be of the form $\bar{c}c$ and $(\bar{c}c)^{-1}$, where $c = w_{i}(\bar{c})$ for some $i \in I$. Since primitive generators are unique up to reversal by Theorem~\ref{theo:main}, it suffices to show that, for distinct pairs $(\bar{c},i)$ and
$(\bar{c}',i')$, the corresponding $(\ell{+}1)$-tuples $(\bar{c}\, w_i(\bar{c}))$ and $(\bar{c}'\, w_{i'}(\bar{c}'))$ are not
the same up to reversal. Now $\bar{c}\, w_i(\bar{c}) = \bar{c}'\, w_{i'}(\bar{c}')$ implies $\bar{c} = \bar{c}'$,  whence $i$ and $i'$ are distinct, whence $w_i(\bar{c}) \neq w_{i'}(\bar{c}')$ by ({w1}), a contradiction. 
On the other hand if $\bar{c}\, w_i(\bar{c}) = (\bar{c}'\, w_{i'}(\bar{c}))^{-1}$, then $\bar{c}'=w_{i}(\bar{c}), c_\ell \cdots c_2$,
whence $w_{i'}(\bar{c}')$ does not occur in $\bar{c}$ by {(w3)}, again a contradiction.

At this point, we have assigned a collection of tuples with primitive length $\ell{+}1$ to the extensions of predicates in $\sigma$ so as to guarantee that  
$\fC^+ \models \forall \bx_{\ell} \exists x_{\ell+1}\, \gamma_i$
for all $i \in I$. In addition, no adjacent $(\ell{+}1)$-types thus defined violate $\forall \bx_{\ell+1}\, \delta$.
It remains to complete the specification of $\fC^+$ by defining the adjacent types of all remaining primitive $\ell{+}1$-tuples, 
and showing that, in the resulting structure, every $(\ell{+}1)$-tuple (primitive or not) satisfies $\delta$.
Let $c_1 \cdots c_{\ell+1}$ be a primitive $(\ell{+}1)$-tuple whose adjacent type has not yet been defined. Let
$\zeta = \atp^\fC[c_1 \cdots c_{\ell}]$ 
and $\eta= \atp^\fC[c_2 \cdots c_{\ell+1}]$. 
Writing $c_h = (b_h, i_h, j_h)$ for all $h$  ($1 \leq h \leq \ell+1$),
we have $\zeta = \atp^{\fB^-}[b_1 \cdots b_{\ell}]$ and $\eta= \atp^{\fB^-}[b_2 \cdots b_{\ell+1}]$.
By~\eqref{eq:reduction1}, $\fB \models p_\zeta[b_2 \cdots b_{\ell+1}]$, and hence by~\eqref{eq:reduction3},
$\fB \models (\zeta \wedge \hat{\delta})^\circ[b_2 \cdots b_{\ell+1}]$, whence 
$\zeta  \wedge \eta^+ \wedge \hat{\delta}$ is consistent, by the definition of the operator~$(\cdot)^\circ$.
So let $\omega \in \Atp^\sigma_{\ell{+}1}$ entail this formula, and
set $\atp^{\fC^+}[\bar{c}c] = \omega$. Carrying this procedure out for all remaining primitive $\ell{+}1$-tuples, we obtain a layered structure $\fC^+$ of primitive length $\ell+1$. Let $\bar{d}$ be any $(\ell{+}1)$-tuple
of elements from $C$.
If $\bar{d}$ is primitive, then we have just ensured that $\fC^+ \models \delta[\bar{d}]$. If, on the other hand, 
$\bar{d} = \bar{e}^{f}$ for some $k$-tuple $\bar{e}$ and some $f \in \bA^{\ell+1}_{k}$, where $k \leq \ell$, then, since
$\phi^\#$, we have $\fC \models \delta^f[\bar{e}]$ and hence, by Lemma~\ref{lma:essentiallyTrivial}, 
$\fC \models \delta[\bar{d}]$. 
This completes the construction of $\fC^+$. We have shown that $\fC^+ \models \phi$.
\end{proof}

Lemma~\ref{lma:base} establishes the decidability of satisfiability for $\AFv{3}$.
Lemmas~\ref{lma:reductionDirection1} and~\ref{lma:reductionDirection2}, on the other hand, reduce the
satisfiability problem for $\AFv{\ell+1}$ to that for $\AFv{\ell}$~($\ell \geq 3$), though with exponential blow-up.
Putting these together, we obtain the decidability of satisfiability for the whole of $\AF$.
More precisely:
\begin{theorem}
If $\phi$ is a satisfiable \AFv{\ell+1}-formula, with $\ell \geq 2$, then 
$\phi$ is satisfied in a structure of size at most $\ft(\ell{-}1, O(\sizeof{\phi}))$.
Hence the satisfiability problem for \AFv{\ell} is in $(\ell{-}2)$-\NExpTime{} for all $\ell \geq 3$, and the adjacent fragment is $\Tower$-complete.
\label{theo:upperBound}
\end{theorem}
\begin{proof}
	Fix $\ell \geq 2$ and suppose $\phi$ is a satisfiable $\AFv{\ell+1}$-formula 
	over a signature $\sigma$.
	By Lemma~\ref{lma:anf}, we may assume that $\phi$ is in normal form.
	%Without loss of generality, we assume $|\sigma| \geq 2$. 
	Writing $\phi_{\ell+1}$ for $\phi$, let $\phi_{\ell}$ be the formula $\phi'$ given by the
	conjunction of $\phi^\#$ and formulas~\eqref{eq:reduction1}--\eqref{eq:reduction3}
	as described before Lemma~\ref{lma:reductionDirection1}. 
	Repeating this process, we obtain a
	sequence of formulas $\phi_{\ell+1}, \dots, \phi_3$. 
	By Lemma~\ref{lma:reductionDirection1}, $\phi_3$ is satisfiable.
	%Since our goal is to obtain a small model property for the language $\AFv{\ell+1}$, we may regard $\ell$ as a constant. 
	For all $k$, ($3 \leq k \leq \ell+1$), let $\phi_k$ have signature $\sigma_k$, and for $k \leq \ell$, consider the construction of $\phi_{k}$ from
	$\phi_{k+1}$. Since $\sum_{k'=1}^{k+1} |\bA^{k+1}_{k'}|$ 
	is bounded by a constant, we see that $\sizeof{\phi^{\#}_{k+1}}$ is $O(\sizeof{\phi_{k+1}})$. 
	Turning now to the formulas corresponding to~\eqref{eq:reduction1}--\eqref{eq:reduction3},
	we employ the same technique used in the proof of Lemma~\ref{lma:smallConnector}.
	When considering the adjacent $k$-types over $\sigma_{k+1}$, we may disregard all adjacent atoms whose argument sequence is not
	a substitution instance of some argument sequence $\bx_{k}^g$ occurring in an atom of $\phi_{k+1}$, as these cannot affect the
	evaluation of $\phi_{k+1}$. And since $k \leq \ell$, the 
	number of functions from $\bx_k$ to itself is again bounded by a constant, so that the number of adjacent $k$-atoms over $\sigma_{k+1}$ 
	that we need to consider is $O(\sizeof{\phi_{k+1}})$. Thus, the number of adjacent $k$-types over $\sigma_{k+1}$ that we need to consider
	is $2^{O(\sizeof{\phi_{k+1}})}$; and this bounds the number of conjuncts in~\eqref{eq:reduction1}--\eqref{eq:reduction3} taken together. 
	Some care is needed when calculating the sizes of these conjuncts themselves, as they feature the subformulas
	$(\zeta \wedge \hat{\delta} \wedge \gamma_i)^\circ$ and
	$(\zeta \wedge \hat{\delta})^\circ$. However, these are simply, in effect, disjunctive normal forms over atoms contained in $\phi_{k+1}$, and hence
	have cardinality $2^{O(\sizeof{\phi_{k+1}})}$, whence $\sizeof{\phi_k}$
	is $2^{O(\sizeof{\phi_{k+1}})}$. By an easy induction, then, $\sizeof{\phi_{3}}$ is $\ft(\ell-2, O(\sizeof{\phi_{\ell+1}}))$, i.e.~$\ft(\ell-2, O(\sizeof{\phi}))$.  
	
	By Lemma~\ref{lma:base}, $\phi_3$ has a model of cardinality $\ft(\ell-1, O(\sizeof{\phi}))$. Moreover, 
	by Lemma~\ref{lma:reductionDirection1}, each of the formulas $\phi_k$ ($3 \leq k \leq \ell$)  
	has a model over a set, say $B_k$, such that  $|B_{k+1}| \leq |B_{k}| \cdot \sizeof{\phi_{k+1}} \cdot (\ell^2 +  \ell + 1)^\ell$. 
	Since $\sizeof{\phi_{k}}$ is $\ft(\ell-3, O(\sizeof{\phi}))$ for all $k$
	and remembering that $\ell$ is a constant, we see that $\phi = \phi_{\ell+1}$ has a model of cardinality 
	$O(\sizeof{\phi}^{\ell-1}\ft(\ell-1, O(\sizeof{\phi})))$, that is to say $\ft(\ell-1, O(\sizeof{\phi}))$.
\end{proof}

%% file: sections/guarded-fragment.tex
%!TEX root = ../main.tex

\section{The Guarded Subfragment}\label{sec:guarded}
We next shift our attention to the \emph{guarded subfragment} of the adjacent fragment, 
denoted~$\GA$, defined as the intersection of the guarded fragment $\GF$
and $\AF$.
Recall that in $\GF$, quantification is relativized by atoms, e.g. all universal quantification 
takes the form $\forall \bar x (\alpha \rightarrow \psi)$, where $\alpha$ (a \emph{guard}) 
is an atom featuring all 
the variables in $\bar{x}$ and all the free variables~of~$\psi$. 
We show that the satisfiability problem for $\GA$, 
in contrast to $\GF^2$ (the two-variable guarded fragment), 
is $\TwoExpTime$-complete, and thus as difficult as full $\GF$.

\input{sections/lower_bound_GAF_extended}

%% file: sections/lower_bound_GAF_extended.tex
Our proof employs the same strategy as the $\TwoExpTime$-hardness 
proof for $\GF$ by Gr\"adel~\cite{Gradel99}.
The novel part of the reduction here concerns a feature characteristic of 
hardness results for guarded logics~\cite{Gradel99,Kieronski19}. However,
the fact that we are working in the guarded \textit{adjacent} fragment means that 
existing techniques are not directly available. 

Let $m \in \N$ and consider the following adjacent functions 
(the upper index is mapped to the lower~one):
\begingroup % nasty hack
\setlength\arraycolsep{3pt}
\begin{align*}
    \lambda_{1} &\coloneqq 
        \begin{pmatrix}
        1 & 2 & 3 & 4 & \dots & m{+}2\\
        1 & 2 & 2 & 3 & \dots & m{+}1
        \end{pmatrix},
    \lambda_{2} \coloneqq 
        \begin{pmatrix}
        1 & 2 & 3 & 4 & \dots & m{+}2\\
        1 & 2 & 1 & 2 & \dots & m
        \end{pmatrix},
    \lambda_{3} \coloneqq 
        \begin{pmatrix}
        1 & 2 & 3 & 4 & \dots & m{+}2 \\
        1 & 2 & 3 & 3 & \dots & m{+}1
        \end{pmatrix}.
\end{align*}
\endgroup
We show that repeated application of $\lambda_1$--$\lambda_3$ on the bit-string $\bfz \bfo \bfo^m$ yields the whole of
$\bfz \bfo \{ \bfz, \bfo \}^m$.
\begin{lemma} \label{lemma_GA_generation_1}
    Let $W_0 \subseteq \{ \bfz, \bfo \}^*$ contain $\bfz \bfo \bfo^m$ and $W_i \coloneqq W_{i{-}1}$
    $\cup \{ \bar{w}^{\lambda_1}, \bar{w}^{\lambda_2}, \bar{w}^{\lambda_3} \mid \bar{w} \in W_{i{-}1} \}$.
    Setting $W \coloneqq \textstyle \bigcup_{i\geq0} W_i$, we have $\bfz \bfo \{ \bfz, \bfo \}^m \subseteq W$. %Then $\bfz \bfo \{ \bfz, \bfo \}^n \subseteq \textstyle \bigcup_{i\geq0} W_i$.
\end{lemma}   
\begin{proof}
We inductively prove that, for any $i \in [0, m]$ 
and any $\bar{c} \in \{ \bfz, \bfo \}^i$, we have $\bfz \bfo \bar{c} \bfo^{m{-}i} \in W$.
The base case ($i = 0$) follows from the assumption that $W_0$ contains $\bfz\bfo\bfo^n$, so let~$i > 0$. 
We aim at generating any word $u = \bfz \bfo x \bar{c} \bfo^{m{-}i{-}1}$ for $x\in \{ \bfz, \bfo \}$.
By induction hypothesis both $\bar{v} = \bfz \bfo \bar{c}\bfo^{m{-}i{-}1} \bfo$ and
$\bar{w} = \bfz \bfo \bar{d} \bfo^{m{-}i{-}1} \bfo \bfo$ (where $\bar{c}= c_1\bar{d}$) are in $W$.
We consider cases: 
(i) if $x = \bfo$ then $\bar{u} = \bar{v}^{\lambda_1}$, 
(ii) if both $x$ and $c_1 = \bfz$ then $\bar{u} = \bar{v}^{\lambda_3}$,
and otherwise, (iii) $x=\bfz$, $c_1 =  \bfo$ and $\bar{u} = \bar{w}^{\lambda_2}$.
Thus~$\bar{u} \in W$.
\end{proof} 

Let $G_{m}$ and $P$ be, respectively, $(m+2)$-ary and binary~predicates.
We define $\zeta^P_{m}$ to be the sentence~below:
\[
    \forall x y \; \Big( P(x y) \to
        G_{m}(x y \underbrace{y \cdots y}_m) \Big) \land
    \bigwedge_{i=1,2,3} \forall \bu_{m+2} \; 
        \Big( G_{m}(\bu_{m+2}) \to
        G_{m}(\bu_{m+2}^{\lambda_i}) \Big).
\]
Let $\str{A}$ be a model of $\zeta^P_{m}$, and take any
$(a,b) \in P^{\str{A}}$. By Lemma~\ref{lemma_GA_generation_1} we conclude 
that $G_m^{\str{A}}$ contains every word of the form $a b \{a, b\}^m$.
Let $R$ be some $4$-ary relation symbol.
In the forthcoming proof we also consider a $(2m + 4)$-ary predicate $F_m$
described by $\epsilon^{R}_{m}$ which is a conjunction of the following two sentences:
\[
    \forall y x x' y' \; \Big( R(y x x' y') \to
        F_m(\underbrace{y \dots y}_m y x x' y' \underbrace{y' \dots y'}_m) \Big)
\]
\[
    \bigwedge_{i, j \in \{0,1,2,3\}} \forall \bu_{m+2}^{-1} \bv_{m+2} \; 
        \Big( F_m(\bu_{m+2}^{-1} \bv_{m+2}) \to
        F_m(z_{\lambda_{i}(m+2)} \dots z_{\lambda_{i}(1)} z'_{\lambda_{j}(1)} \dots z'_{\lambda_{j}(m+2)}) \Big).
\]
Here $\lambda_0$ is the identity function (i.e. $k \mapsto k$ for each $k \in [1, m]$).
The intended meaning is that whenever $\fA \models \zeta^{R}_{m}$ holds, 
this implies that for any quadruple $baa'b' \in R^{\str{A}}$
we have~that $\str{A} \models G[\bar{c} b a a' b' \bar{c}']$ holds for all
$\bar{c} \in \{a, b\}^m$ and $\bar{c'} \in \{a', b'\}^m$.

\textbf{ATMs.} An \emph{Alternating Turing Machine} (ATM) $\tm$
is a tuple $\langle \tmst, \tmal, \tmtrl, \tmtrr, q_0, \kappa \rangle$, 
where $\tmst$ is a finite set of states,
$\tmal$ is a finite alphabet containing an empty-cell symbol `$\blank$',
$\tmtrl, \tmtrr : \tmst \times \tmal \to \tmst \times \tmal \times [-1, 1]$ are, 
respectively, the left and right transition functions,
$q_0 \in Q$ is the initial state, and
$\kappa: \tmst \rightarrow \{ \tmuni, \tmexi \}$ is the state 
descriptor function stating if a given state is, respectively, 
\emph{universal} $\tmuni$ or \emph{existential} $\tmexi$.
We say that a state $q$ is \emph{accepting} if and only if $\kappa(q) = \tmuni$, and there are no possible transitions given by $\tmtrl(q, s)$ and $\tmtrr(q, s)$ for any symbol $s \in \tmal$.
Dually, $q$ is \emph{rejecting} if and only if $\kappa(q) = \tmexi$, and, again, there are no possible transitions.
%Given a word $\bar{\omega} \in \tmal^*$ we call $\tm(\bar{\omega})$ a \emph{run of $\tm$ on input $\bar{\omega}$}.
By tracking every step of the computation by $\tm$ on $\bar{w} \in \tmal^*$
we obtain a binary tree structure $\tmct = \langle \tmctver, \tmctverl, \tmctverr \rangle$. Each vertex $v \in \tmctver$ is labelled with some %\ianside{No. Each node is labelled with a tuple, not regarded as one.}
{\em configuration} $\langle q,\bar{s},h\rangle$, where
$q \in \tmst$ is a machine state, $\bar{s} \in \tmal^m$ a word indicating the contents of the tape,
and $h \in [0, m-1]$ an integer indicating the position of the head at some point in the computation.
Each edge $(v, u) \in \tmctvera$ (where $\eta=l,r$) is labelled by a transition $\tmtra(q, s_h)$, enabled in the configuration~labelling~of~$v$.
We call $\tmct$ the \emph{configuration tree} of the computation by $\tm$ on $\bar{w}$. 
Assuming $\tmct$ is finite,
we say that it is \emph{accepting} if no vertex is associated with a rejecting state.
We identify the following three properties of $\tmct$.
The root node is labelled with a configuration in which the machine is in state $q_0$, the head position is 
$0$, and the tape is written with the string $\bar{w}$ followed by blanks.
We call this property {\em initial configuration} (\textbf{IC}).
Let $u$ be a vertex labelled with a configuration in which the machine is in state $q$.
If $q$ is universal, then $u$ will have two children; if $q$ is existential, then $u$ will have a single child; if $q$ is accepting, or rejecting then $u$ will have no children.
We call this property {\em successor existence} (\textbf{SE}).
Suppose further that, in the configuration labelling of some node $u$ we have that the head is reading the symbol $s$ whilst in state $q$. Then any child
$v_\eta$ (s.t. $(u, v_\eta) \in \tmtra$, where $\eta=l$ or $\eta = r$) represents the result of a single transition $\tmtra(q, s) = (p, s', k)$, and thus is labelled with a configuration in which the machine state
is $p$, $s'$ is written in place of $s$, and the head is moved by a distance of $k$.
We call this property \textit{configuration succession}~(\textbf{CS}).

\textbf{Encoding numbers.} 
Let $\bin{m}{x, y}$ be the canonical~map~from $[0, 2^n - 1]$ to
bit-string representations of length $m \in \N$, using $x$ as the \emph{zero bit} and $y$ as the \emph{unit bit}.
In the sequel, we will consider structures $\str{A}$ with elements labelled by a unary predicate $O$. 
If $\fA \models \lnot O[a] \land O[b]$ we say that $a, b$ \emph{act as zero and unit bits}. 
We thus associate every word $\bar{c} \in \{a,b\}^+$ with an
integer value given by the canonical map $\val{\fA}$ (this function depends on $\fA$, because it is $\fA$ that
determines which is the zero bit and which is the unit bit.)
Given two bit-strings $\bar{c}$ and $\bar{d}$ (not necessarily composed of the same elements) there is a classical 
way to define the following properties in the monadic fragment of $\FO$ (hence also in $\GA$):
\begin{itemize}\itemsep0em
\item $\fA \models \shless{\bar{c}}{\bar{d}}$ iff $\val{\fA}(\bar{c}) < \val{\fA}(\bar{d})$ \label{GAF_less_def}
\item $\fA \models \sheq{\bar{c}}{\bar{d}}$ iff $\val{\fA}(\bar{c}) = \val{\fA}(\bar{d})$
\item $\fA \models \sheq{\bar{c}}{\bar{d} + k}$ iff $\val{\fA}(\bar{c}) = \val{\fA}(\bar{d}) + k$, where $k \in [-1, 1]$.
\end{itemize}
Formally, the formulas are defined as follows:
\[
    \shless{\bu_m}{\bv_m} \coloneqq \bigvee_{i=1}^m \Big( \neg O(z_i) \wedge O(z'_i) \wedge 
    \bigwedge_{j = i + 1}^{m} \big( O(z_j) \leftrightarrow O(z'_j) \big) \Big)
\]
\[
    \sheq{\bu_m}{\bv_m} \coloneqq \bigwedge_{i = 1}^{m} \Big( O(z_i) \leftrightarrow O(z'_i) \Big)    
\]
\[
    \sheq{\bu_m}{\bv_m + 1} \coloneqq \bigwedge_{i = 1}^{m} \Big( \big( O(z_i) \leftrightarrow O(z'_i) \big) \leftrightarrow \bigvee_{j=1}^{i-1} O(z_j) \Big)
\]
and where $\sheq{\bu_m}{\bv_m + 0} \coloneqq \sheq{\bu_m}{\bv_m}$ and $\sheq{\bu_m}{\bv_m - 1} \coloneqq \sheq{\bv_m}{\bu_m +~1}$.

Fix an ATM $\tm$
working in exponential space w.r.t any given input $\bar{w}$.
Our goal is to construct a polynomial-size $\GA$-sentence $\varphi_{\tm, \bar{w}}$
which is satisfiable if and only if $\tm$ has an accepting configuration tree on a given input $\bar{w}$.
Utilising the fact that $\AExpSpace$ equals $\TwoExpTime$, the reduction yields the desired bound on $\GA$.
Now, take an accepting configuration tree 
$\tmct = (\tmctver, \tmctverl, \tmctverr)$ for $\tm$~and~$\bar{w}$, and fix $n =|\bar{w}|$.
We consider structures interpreting binary predicates $V, R, Q_q$ (for each $q \in \tmst$), quaternary 
predicates $E_l, E_r$ and $n$-ary predicates $H, S_s$ (for each $s \in \tmal$).
We say that $\str{A}_0$ \emph{embeds $\tmct$} if there is $f: \tmctver \to A_0^2$ such that for all $v \in \tmctver$
\begin{enumerate}[label=(\alph*)]\itemsep0em
    \item $\fA_0 \models V[f(v)]$, \label{GAF:list_model_A_first}
    \item $\fA_0 \models R[f(v)]$ if $v$ is the root node,
    \item $\fA_0 \models E_\eta[f(u)^{-1}, f(v)]$ if $(u, v) \in \tmctvera$, where $\eta=l,r$,
    \item $\fA_0 \models Q_q[f(v)]$ if the configuration $v$ is in state $q$,
    \item $\fA_0 \models S_s[\bin{n}{f(v)}(i)]$ if $v$'s $i$-th tape cell has symbol $s$, 
    \item $\fA_0 \models H[\bin{n}{f(v)}(i)]$ if $v$'s head is located over the $i$-th tape square. \label{GAF:list_model_A_last}
\end{enumerate}
We construct $\str{A}_0$ embedding $\tmct$ as follows:
the domain $A_0$ is composed of fresh symbols $0_v, 1_v$ for 
each vertex $v \in \tmctver$, for which we also put $f(v) = 0_v1_v$. (Notice that $0_v1_v$ is a word over $A_0$
of length 2.)
We interpret the predicates $V, R, E, Q_q, S_s$ and $H$ as required 
by conditions \ref{GAF:list_model_A_first}--\ref{GAF:list_model_A_last}.
Then, we construct a $\varphi_{\tm, \bar{w}}$ in $\GA$ such that: 
(i) $\str{A}_0$ can be expanded to a model $\str{A}$ of $\varphi_{\tm, \bar{w}}$;
and (ii) every model of $\varphi_{\tm, \bar{w}}$ embeds~$\tmct$.

The first conjunct of $\varphi_{\tm, \bar{w}}$ requires pairs $ab$ satisfying the predicate $V$ to act as zero bits and unit bits,
indicated by the predicate $O$:
\[
    \varphi_1 = \forall x y \; \Big( V(x y) \to \neg O(x) \land O(y) \Big)
\]
%
%Given a pair of elements $a, b \in A$ satisfying $\fA \models \neg O[a]$ 
%and $\fA \models O[b]$, we treat these elements, respectively, as turned-off 
%and turned-on bits of the binary encoding of integers. 
%Thus for any sequence $\bar{u} \in \{a,b\}^+$ we can associate the 
%integer value given by the canonical map $\val{\fA}$.

We now add $\zeta^V_n$ and $\zeta^V_{2n}$ to the main formula $\varphi_{\tm, \bar{w}}$.
Recall that the sentence $\zeta^V_m$ features an $(m+2)$-ary predicate $G_m$, and ensures that,
if $\fA \models V[a b]$, then $\fA \models G_m[a b \bar{c}]$ for all $c \in \{a, b\}^m$.
Writing $\varphi_2$, $\varphi_3$ and $\varphi_4$ as
\[ 
    \bigwedge_{p, q \in \tmst}^{p \neq q}
        \forall x y \Big( V(x y) 
        \to \big( \neg Q_{p}(x y) \lor \neg Q_{q}(x y) \big) \Big),
\]
\[ 
    \bigwedge_{s, s' \in \tmal}^{s \neq s'}
       \forall x y \bu_n \Big(
           G_n(x y \bu) {\to} \lnot \big( 
                S_s(\bu_n) {\land} S_{s'}(\bu_n) \big)
       \Big),
\]
\[
    \forall xy\bu_n\bv_n \bigg(
        G_{2n}(x y \bu_n \bv_n) \to \Big( 
            \big( H(\bu_n) \land H(\bv_n) \big) \to \sheq{\bu_n}{\bv_n}
        \Big) \bigg)
\]
respectively, we ensure that every configuration is in at most one state at a time,
every tape square of a configuration has at most one symbol, and the read-write head of any configuration is pointing to a single square at a time.
Note that all of these formulas are guarded. However, the advertised behaviour of the guard predicates $G_n$ and $G_{2n}$ means,
in essence, that the guards have no semantic effect.

We now secure the property (\textbf{IC}). Let $\mu_1$ abbreviate the formula
$Q_{q_0}(x y) \land H(\bin{n}{x y}(0)) \land \bigwedge_{i = 1}^{|\bar{w}|} S_{w_i}(\bin{n}{x, y}(i{-}1))$, and
$\mu_2$ the formula $\shless{\bin{n}{x, y}(|\bar{w}| {-} 1)}{\bu_n} \to S_{\blank}(\bu_n)$. Writing 
\[
    \varphi_5 \coloneqq \exists xy \Big( V(x y) \wedge R(x y) \Big) \wedge
    \forall xy \Big( R(x y) \to
    \big( \mu_1 \wedge \forall \bu_n (G(x y \bu_n) \to \mu_2) \big) \Big),
\]
we ensure that there is a root configuration in which
the machine state is $q_0$, the head is scanning square `0', and the tape is written with the string $\bar{w}$ followed by the requisite number of blanks.

Let $\text{K}_{\tmuni}$ be the formula $\textstyle \bigvee^{\kappa(q) {=} \tmuni}_{q \in \tmst} Q_q(xy)$, 
and define $\text{K}_{\tmexi}$ analogously.
Similarly, we define $\text{K}_{\tmrej}$ to be a disjunction of rejecting states.
The formula $\varphi_6 \coloneqq \forall x y ( V(x y) \to \neg \text{K}_{\tmrej}(x y) )$ 
ensures that no configuration is labelled with a rejecting state.%\ianside{There's a weakness here. Some conjucnts enforce properties SE, IS or whatever. Others are just hanging around. This is uneven. Not sure if this can be fixed now}

We next encode the transitions of $\tm$, securing the property (\textbf{SE}).%\bbeside{do we need to say anything about acc state?}\ianside{See earlier remarks about states with no successors}
Let 
$\psi_{\tmuni}$ abbreviate the formula $\exists x' y' E_l(y x x' y') \land \exists x' y' E_r(y x x' y')$, and
$\psi_{\tmexi}$ the formula $\exists x' y' E_l(y x x' y') \lor \exists x' y' E_r(y x x' y') $.
Writing
\[
\varphi_7 \coloneqq \bigwedge_{k=\tmuni, \tmexi} \forall y x \bigg( V(x y) \to
\Big( \text{K}_{k}(x y) \to \psi_k \Big) \bigg).
\]
we ensure that, if $\fA \models V[a, b] {\land} \text{K}_{k}[a b]$, then $\fA$ contains 
pairs encoding the appropriate successor configurations.

We next ensure that the transitions have the expected effect on the configurations they connect, securing
the property (\textbf{CS}). For this, we
need a further predicate, $F_n$, to act as a dummy guard.
By adding $\epsilon^{E_\eta}_{n}$ to the main formula (for both $\eta=l,r$), we secure $\fA \models$ $F_n[\bar{c} b a a' b' \bar{c}']$
for all $a, b, a', b'$ such that $\fA \models E_\eta[b a a' b']$ 
with $\bar{c} \in \{ a, b \}^n$, $\bar{c}' \in \{ a', b' \}^n$.
The formula $\varphi_8$ then ensures that any pair of parent and successor configurations 
have identical tape contents except (possibly) for the position scanned by the head, thus:
\begin{multline*}
    \varphi_8 \coloneqq \forall \bu_n y x x' y' \bv_n \; 
        \Bigg( F_n(\bu_n y x x' y' \bv_n) \to \\
            \bigg( \Big( \neg H(\bu_n) \land \sheq{\bu_n}{\bv_n} \Big) \to 
                \Big( \bigwedge_{s \in \tmal} \big( S_s(\bu_n) \to S_s(\bv_n) \big) \Big) 
        \bigg) \Bigg).
\end{multline*}
Now let $\chi_1$ abbreviate the formula $Q_p(x' y')$,
$\chi_2$ the formula $\sheq{\bv_n}{\bu_n} \to S_{s'}(\bv_n)$, and $\chi_3$ the formula $\sheq{\bv_n}{\bu_n + k} \to H(\bv_n)$.
In addition, we write~$\xi_{\tau_\eta}$~for the sentence %\ianside{I wonder if it wouldn't be better to label these formulas with the transitions.}
\[
    \forall \bu_n y x x' y' \bv_n \; 
	\bigg( G(\bu_n y x x' y' \bv_n) {\to}
	    \Big( \big( E_\eta(y x x' y') \land  Q_q(x y) \land H(\bu_n) \land S_s(\bu_n) \big) {\to}
	    \big( \chi_1 \land \chi_2 \land \chi_3 \big) \Big)
        \bigg).
\]
Assuming the transition $\tau_\eta$ is of the form $(q, s) \mapsto (p, s', k)$, the formula $\xi_{\tau_\eta}$
states that, if in a certain configuration, the machine state is $q$ and the head is reading symbol $s$, 
then in the $\eta$-side successor configuration defined by $\tmtra$, the machine state will be
$p$, the symbol $s$ will have been replaced by $s'$, and  the head will have moved by $k$.
To encode all possible transitions, we write $\varphi_9$ to be a conjunction of $\xi_{\tau_\eta}$
for each transition $\tau_\eta \in \tmtra$ for both~$\eta=l,r$.
%The parameters $p$, $s'$, $k$ are then the state, symbol and head movement from~$\tmtra(q, s)$.\ianside{Yuck. Why don't you define a transition $\tau$ to be a quadruple and define $\xi_\tau$?}

Let $\fA_0$ embed some accepting $\tmct$ as described in \ref{GAF:list_model_A_first}--\ref{GAF:list_model_A_last}.
We expand $\fA_0$ to $\fA$ by setting $A = A_0$ with
\begin{enumerate}
    \item $\fA \models \lnot O[a]$ and $\fA \models O[b]$ if $\fA_0 \models V[ab]$,
    \item $\fA \models G_m[a b \bar{c}]$ where $\fA_0 \models V[a b]$ and $\bar{c} \in \{ a, b \}^m$ (for $m=n, 2n$),
    \item $\fA \models F_n[\bar{c} b a a' b' \bar{c}']$ where $\fA_0 \models S_\eta[b a a' b']$, $\bar{c} \in \{ a, b \}^n$ and
    $\bar{c}' \in \{ a', b' \}^n$ (here $\eta=l,r$).
\end{enumerate}
Recalling that $\tmct$ contains an initial configuration (\textbf{IC}), we have that $\fA \models \varphi_5$.
Additionally, $\tmct$ {has the property} (\textbf{SE}), we see that $\fA \models \varphi_7$.
Lastly, since $\tmct$ has the property (\textbf{CS}), we have that $\fA \models \varphi_8$.
At this point it is easy to verify that $\fA \models \varphi_{\tm, \bar{w}}$.

Conversely, suppose $\fA \models \varphi_{\tm, \bar{w}}$.
We
construct an embedding $f\colon \tmctver {\to} A^2$ for an accepting $\tmct$ by well-founded induction.
The following observations will be used. Suppose $\fA \models V[ab]$. Intuitively, we think of the pair
$ab$ as a vertex of the computation tree labelled by some configuration, as determined by the predicates
$Q_q$, $S_s$ and $H$.
By $\varphi_2$, there is a unique $Q_q$ (for $q \in \tmst$) satisfied by $ab$.
Moreover, by $\varphi_3$, any bit-string $\bar{c} \in \{a, b\}^n$ satisfies a unique $S_s$ (for $s \in \tmal$).
Similarly, by $\varphi_4$, there is a unique $\bar{c}$ satisfying $H$.

Proceeding with the induction,
for the base case, pick $a, b$ s.t. $\fA \models R[ab]$.
By $\varphi_5$ we see that $\fA \models Q_{q_0}[ab]$, $\fA \models H[\bin{n}{a, b}(0)]$, and $\fA \models S_{w_i}[\bar{c}]$ for each $1 \leq i \leq |\bar{w}|$
with $\val{\fA}(\bar{c}) = i-1$  and $\fA \models S_{\blank}[\bar{c}]$ otherwise.
We then set $\tmctver = \{ v \}$ and $f(v) = ab$. Labeling $v$ with the state, tape and head position as suggested by
\ref{GAF:list_model_A_first}--\ref{GAF:list_model_A_last}, we 
have secured the property (\textbf{IC}).

For the inductive step, let $u$ be vertex which has been added to the tree.
Assume $u$ is labelled with a configuration in which the machine state $q$ is universal, and the head is at position $h$, reading symbol $s$.
If $q$ is accepting, then we stop. Otherwise, $\varphi_7$ guarantees that there are words $a_\eta b_\eta \in A^2$
such that $\fA \models S_\eta[f(u)^{-1} a_\eta b_\eta]$ for both $\eta=l,r$.
Notice that by $\varphi_8$ if the configuration labelling $u$ has the symbol $s'$ written on tape square $i$ (for $i \neq h$), then $\fA \models S_{s'}[\bin{n}{a_\eta b_\eta}(i)]$.
By $\varphi_9$ the pair $a_\eta b_\eta$ satisfies the predicate $Q_p$ that is in accordance with the transition $\tmtra(q, s)=(p, s', k)$.
Additionally, $\fA \models S_{s'}[\bin{n}{a_\eta b_\eta}(h)]$ and $\fA \models H[\bin{n}{a_\eta b_\eta}(h + k)]$.

We thus set $\tmctver \coloneqq \tmctver \cup \{ v_\eta \}$, $f(v_\eta) = (a_\eta, b_\eta)$ and $\tmctvera \coloneqq \tmctvera \cup \{ (u, v_\eta) \}$ thus securing (\textbf{SE}).
By interpreting the state, tape and head position of $v_\eta$ as suggested by \ref{GAF:list_model_A_first}--\ref{GAF:list_model_A_last}
we see that $v_\eta$ is a proper successor of $u$ as required by (\textbf{CS}).
The case for when $q$ is existential is similar.

Since there are no rejecting states in conf. tree (reference $\varphi_6$),
there is an initial configuration by (\textbf{IC}),
each parent has children complying with (\textbf{SE}),
and each parent-child pair conforms to (\textbf{CS}), we conclude that $\tmct$ is an accepting configuration tree.

%% file: sections/conclusions.tex
%!TEX root = ../main.tex

\newcommand{\posetP}{\mathcal{P}}

\section{Conclusions}\label{sec:extensions-and-future-work}

The adjacent fragment $\AF$ is defined as the union of the formulas sets $\AFv{[k]}$, each of which restricts the allowed argument sequences appearing in atomic formulas to \textit{adjacent} words over the alphabet $\bx_k$.
The question arises as to whether these restrictions might be further relaxed without 
compromising the decidability of satisfiability. Under reasonable assumptions about the fragment in question, the answer must be no. Indeed, assume, for simplicity, that the argument
sequence $x_1x_2$ is allowed in the 2-variable case, and $x_2x_3$ in the 3-variable case. Now the only non-adjacent words of length 2 over $\bx_3$ are $x_1x_3$ and $x_3x_1$. In the first case, this allows us to write the formula $\forall x_1 \forall x_2(r(x_1x_2) \rightarrow \forall x_3(r(x_2x_3) \rightarrow r(x_1x_3)))$, which says that $r$ is transitive. 
But even two-variable logic with (at least two) transitive relations yields a logic for which satisfiability and finite satisfiability are undecidable~\cite{Kieronski-csl12},
since it is simple to write formulas all of whose models embed grids of unbounded size. 
Similar remarks apply to the formula $\forall x_1 \forall x_2(r(x_1x_2) \rightarrow \forall x_3(r(x_2x_3) \rightarrow r(x_3x_1)))$, and indeed to the case of formulas featuring ternary non-adjacent atoms such as $p(x_1x_3x_2)$. 
It is therefore difficult to conceive of meaningful fragments of first-order logic defined purely by reference to restrictions on the allowable argument sequences that  do not define sub-fragments of the adjacent fragment, and that are at the same time decidable for satisfiability. In this respect, \AF{} appears to be the end of the road.

On the other hand, the last two decades have witnessed concerted attempts to investigate the decidability of the satisfiability problem for $\FO^2$ 
over various classes of structures, where 
some distinguished predicates are 
interpreted in a~special way, e.g. as linear orders~\cite{KieronskiT09, SchwentickZ10}; other such semantic constraints have also been investigated~\cite{Manuel10, CharatonikW16, KieronskiMPT14, BojanczykDMSS11, CharatonikKM14, BednarczykCK17}.
It is therefore natural to ask whether the adjacent fragment remains decidable when subject to similar
semantic constraints.
Of course, since $\AF$ extends $\FO^2$, all the undecidability results for $\FO^2$  immediately transfer to $\AF$.
Thus, $\AF$ extended with two transitive relations~\cite{Kieronski-csl12}, or with three equivalence
relations~\cite{KieronskiO12}, or with
one transitive and one
equivalence relation~\cite{KieronskiT09},
or with two linear orders and their two
corresponding successor relations~\cite{Manuel10}, must all be undecidable.
(See~\cite{KieronskiPT18}~for~a~survey.) Regarding positive results, existing results on the fluted fragment give cause for some hope. Thus, for example, the fluted fragment remains decidable with the addition of one transitive relation  (and equality)~\cite{Pratt-HartmannT22}; moreover, \textit{finite} satisfiability for $\FO^2$ with one transitive relation is also 
known to be decidable~\cite{Pratt-Hartmann18}.

A second generalization of $\FOt$ which preserves decidability of satisfiability is the extension with counting quantifiers~\cite{Pratt-Hartmann10,Pratt-Hartmann21,BenediktKT20}. (Here, however, the finite model property is lost.) It has been shown that the corresponding extension of the fluted fragment retains the finite model property~\cite{Pratt-Hartmann21}. Extending the adjacent fragment with counting quantifiers certainly results in loss of the finite model property, because $\AF$ includes $\FO^2$; however, the decidability of the satisfiability and finite satisfiability problems~is~left~for~future~work.

%% file: sections/appendix-preliminaries.tex
%!TEX root = ../main.tex

\section{Appendix for Section~\ref{sec:preliminaries}}
The goal of this section is to prove Lemmas~\ref{lemma:fo2-and-af-over-binary-sig-are-the-same} and Lemma~\ref{lma:simpleComb}. For the former, we 
first prove an auxiliary result:
\begin{lemma}
	\label{lemma:mini-lemma-for-FO-with-two-free-vars}
    Let $\varphi$ be some $\FO$ formula over a relational signature. If no sub-formula of $\varphi$ contains more than $2$ free variables, then, after variable renaming and elimination of vacuous quantification, $\varphi$ is
    in $\FO^2$.
\end{lemma}
\begin{proof}
    Suppose $\varphi$ satisfies the conditions of the lemma, and all vacuous quantification has been eliminated. 
    We show by structural induction that every sub-formula $\psi$ of $\varphi$ may be transformed, by 
    renaming bound variables, into a formula $\psi'$ in which only two variables occur (free or bound); the lemma then follows.
    The base case, where $\psi$ is atomic, is immediate, as is the case $\psi \coloneqq \lnot \chi$. If  $\psi \coloneqq \chi_1 \land \chi_2$, then
    the free variables of both $\chi_1$ and $\chi_2$ are by assumption contained some set $\set{u,v}$. By inductive hypothesis, let 
    $\chi'_1$ and $\chi'_2$ be \AF-formulas obtained by renaming bound variables in $\chi_1$ and $\chi_2$, with
    $\chi'_1$ and $\chi'_2$ having no variables (free or bound) other than $u$ and $v$. Then we may set
    $\psi'\coloneqq \chi'_1 \wedge \chi'_2$.
    Finally, let $\psi\coloneqq \exists u\, \chi$, where $u$ appears free in $\chi$. By inductive hypothesis, let $\chi$ be obtained by renaming bound variables in $\chi$, with $\chi'$ having no variables (free or bound) other than $u$ and $v$. Then we may set $\psi'\coloneqq \exists u\, \chi'$.
\end{proof}

\newtheorem*{RestateLemma_fo2-af-bin-sig}{Lemma~\ref{lemma:fo2-and-af-over-binary-sig-are-the-same}}
\begin{RestateLemma_fo2-af-bin-sig}    
Every $\FO^2$-formula is logically equivalent to an~$\AF$-formula. 
The converse holds for $\AF$-formulas featuring predicates of arity at most two.
\end{RestateLemma_fo2-af-bin-sig}
\begin{proof}
For this proof we innocently augment the syntax of $\AF$ and $\FO^2$ with what we call \textit{units}.
Suppose $\psi$ is some formula. We call $U(\psi)$ \textit{the unit $\psi$}, where $\psi$ itself may contain boolean combinations of units and formulas.
Semantically, we specify units to be inert; i.e. $\fA \models U(\psi) \iff \fA \models \psi$.
Clearly, every formula from a language with units is logically equivalent to a formula without them.
In the sequel, we say that a sub-formula $\theta$ of $\psi$ is contained within a unit if there is a sub-formula $U(\eta)$ of $\psi$ such that $\theta$ is contained within $\eta$.

Now, suppose $\psi$ is a formula every quantified sub-formula of which is contained within a unit. (This includes the case where
$\phi$ is quantifier free.)
We define a procedure for obtaining conjunctive normal form as follows.
(i) Rewrite every $\alpha \leftrightarrow \beta$ that is not contained within a unit to $(\alpha \to \beta) \wedge (\beta \to \alpha)$
and each $(\alpha \to \beta)$ that is not contained within a unit to $(\lnot \alpha \vee \beta)$,
(ii) for each $\lnot \alpha$ such that $\alpha$ is neither a unit nor is contained within a unit, move the negation inwards whilst also removing double negations,
(iii) assuming $\alpha \vee (\beta \wedge \gamma)$ is not contained within a unit, distribute the disjunction inwards over the conjunction.
Disjunctive normal form is defined analogously.

Let $\phi$ be any formula of $\FO^2$ in variables $u, v$ without the unit operator. 
We will preprocess $\phi$ in such a way that the resulting formula $\phi'$ is logically equivalent to $\phi$
and so that for each $y, z \in \{ u, v \}$ with $y \neq z$ and each subformula $\psi_1 \coloneqq \mathsf{Q}_1 y \xi_1$ we have that
$\psi_3 \coloneqq \mathsf{Q}_3 y \xi_3$ appears in $\xi_1$
only if there is some $\psi_2 \coloneqq \mathsf{Q}_2 z \xi_2$ such that $\psi_3$ is a subformula of $\xi_2$.
We call formulas satisfying this condition \textit{properly sequenced}.

To show that every $\FO^2$ is logically equivalent to a properly sequenced $\FO^2$ formula, consider any subformula $\psi \coloneqq \forall y \chi$ of $\phi$
such that $\psi$ is not contained within a unit and $\chi$ features quantifiers only if they are contained within a unit.
Assume, without loss of generality, that $\chi$ is in conjunctive normal form as described above.
We move the universal quantifier inwards and view the $i$-th conjunct of $\chi$ as
$\theta_i \coloneqq \forall y (\alpha(y, z) \lor \beta(y) \lor \gamma(z))$.
Notice that the aforementioned formula is logically equivalent to
$\theta'_i \coloneqq U(\forall y (\alpha(y, z) \lor \beta(y))) \lor \gamma(z)$ as $y$ does not appear free in $\gamma$.
Each $\theta_i$ is now properly sequenced, as each subformula of $\theta_i$ is also properly sequenced.
Define $\psi' \coloneqq \bigwedge_{i} \theta'_i$ and let $\varphi'$ be $\phi'$ but with $\psi$ replaced with $\psi'$.
The case for when $\psi \coloneqq \exists y \chi$ is symmetric (the major difference begin that we turn $\chi$ to disjunctive normal form).
Now, set $\varphi \coloneqq \varphi'$ and run the procedure again until the resulting formula is contained within a unit.

Let $\phi$ be a properly sequenced $\FO^2$ formula (without the unit operator).
We will again assume that $y, z \in \{ u, v \}$ with $y \neq z$.
Now, for all $i \in \N$ denote by $\psi(x_i, x_{i+1})$ the result of substituting all free occurrences of $y$ and $z$
by $x_i$ and $x_{i+1}$ respectively. We claim that for each formula $\phi$ in $\FO^2$ there is logically equivalent formula in $\AF^{[i+1]}$.
Without loss of generality, assume that $\phi$ is properly sequenced.
We proceed via structural induction. The base case is immediate as every atomic formula with two variables is adjacent.
For the interesting cases, let $\psi \coloneqq \chi_1 \land \chi_2$ be an $\FO^2$ formula. By I.H. both $\chi_1$ and $\chi_2$ have logically equivalent counterparts
$\chi'_1(x_i, x_{i+1})$ and $\chi'_2(x_j, x_{j+1})$ in $\AF^{[i+1]}$ and $\AF^{[j+1]}$ respectively. Assume, without loss of generality, that $i \geq j$.
Thus, set $k = i - j$ and let $\chi''_2$ be the result of incrementing each variable index in $\chi'_2$ by $k$. Clearly, $\chi''_2$ is logically equivalent to $\chi'_2$.
Moreover, $\chi''_2(x_i, x_{i+1})$ is in $\AF^{[i+1]}$. Then, the formula $\psi' \coloneqq \chi'_1(x_i, x_{i+1}) \land \chi''_2(x_i, x_{i+1})$ is in $\AF^{[i+1]}$ and logically equivalent
to $\psi$ as required.
Lastly, take any $\FO^2$ formula $\psi \coloneqq \exists y \, \chi$. By proper sequencing we have that each maximal sub-formula of the form $\theta \coloneqq \exists w \, \eta$ in $\chi$ has $w \neq y$;
we can thus not worry about immediate requantifications of variables.
By I.H. $\chi$ is logically equivalent to an $\AF^{[i+1]}$ formula $\chi'(x_i, x_{i+1})$. %(and each $\eta$ is equivalent to a formula $\eta(x_{i+1}, x_{i+2})$).
We thus define $\psi' \coloneqq \exists x_{i+1} \chi'(x_i, x_{i+1})$ which is in $\AF^{[i]}$ as required.

For the converse direction take $\varphi$ in $\AF$ over a signature of unary and binary predicates.
By Lemma~\ref{lemma:mini-lemma-for-FO-with-two-free-vars},
we need only show that $\varphi$ is logically equivalent to a $\AF$ formula 
containing no sub-formula with more than 2 free variables.
Now let $\psi \coloneqq \forall x_{n+1} \chi$ be a sub-formula of $\varphi$, where every occurrence of a quantifier in $\chi$ is within a unit. (This includes
the case where $\chi$ is quantifier-free.)
Convert $\chi$ to conjunctive normal form as described above. Each conjunct of $\chi$ 
is a disjunction of adjacent literals featuring a predicate of arity at most 2. By separating out those literals featuring 
the variable $x_{n+1}$, we may write such a conjunct in the form
$\bigvee_{j} \alpha_j \lor \bigvee_{k} \beta_k$, where
the variables of each $\alpha_j$ are among $x_1, \dots, x_{n}$, and the variables of each
$\beta_k$ are among $x_n, x_{n+1}$ and include $x_{n+1}$.
Moving the universal quantification into the conjunction, we obtain conjuncts of the form $\forall x_{n+1} (\bigvee_{j} \alpha_j \lor \bigvee_{k} \beta_k(x_{n}, x_{n+1}))$; and
since no $\alpha_j$ contains the variable $x_{n+1}$, we move the quantifier inwards once more: $\bigvee_{j} \alpha_j \lor \forall x_{n+1} \bigvee_{k} \beta_k(x_{n}, x_{n+1})$.
Writing $\forall x_{n+1} \bigvee_{k} \beta_k(x_{n}, x_{n+1})$ as $\theta(x_n)$, we conclude that no sub-formula of $\theta(x_n)$ contains more than two variables.
Now, replace $\theta(x_n)$ with the formula $U(\theta(x_n))$ and repeat the procedure.
The case for when $\psi \coloneqq \exists x_{n+1} \chi$ is symmetric (the major difference begin that we turn $\chi$ to disjunctive normal form).

After the procedure described above has run its course, we replace every $U(\psi)$ with $\psi$.
We will thus be left with $\varphi'$ (that is logically equivalent to $\varphi$) which has no sub-formula with more than two free variables.
By Lemma \ref{lemma:mini-lemma-for-FO-with-two-free-vars} we have that $\varphi'$ is logically equivalent to a $\FO^2$ formula which completes the proof. 
\end{proof}

\newtheorem*{RestateLemma_simpleComb}{Lemma~\ref{lma:simpleComb}}
\begin{RestateLemma_simpleComb}
For any integer $k>0$ there is a set $J$ with $|J| = (k^2 + k+1)^{k+1}$ 
and a function $g\colon J^k \rightarrow J$ such that,
for any tuple $\bar{t} \in J^k$ consisting of the elements $t_1, \dots, t_k$ in some order:
{\em (i)} $g(\bar{t})$ is not in $\bar{t}$;
{\em (ii)} if $\bar{t}' \in J^k$ consists of the elements $\set{t_2, \dots, t_k, g(\bar{t})}$ in some order, then
$g(\bar{t}')$  is not in $\bar{t}$ either.
\end{RestateLemma_simpleComb}
\begin{proof}
Take $z = k^2{+}k{+}1$ and $J = [1,z]^{k{+}1}$. 
Thus, $|J|$~is~as required.
Writing any element of $J$ as the word $i\bar{s}$, where $i \in [1,z]$ and $\bar{s} \in [1,z]^k$,
let $g$ be defined by $g(i_1\bar{s}_1, \dots, i_k\bar{s}_k) = i_0 i_1 \cdots i_k$, where $i_0$ is the smallest positive integer
not in the set $S = \set{i_1, \dots , i_k} \cup \bar{s}_1 \cup \cdots \cup \bar{s}_k$. (For
brevity, we are here identifying words over the integers with the sets of their members.) Note that $i_0 \in [1,z]$ since $|S| <z$. 
Now let some tuple $\bar{t} \in J^k$ be given,
consisting of words $t_1, \dots, t_k$ in some order, and write $t_h = i_h \bar{s}_h$ (for all $1 \leq h \leq k$), where $i_h \in [1,z]$ and $\bar{s}_h \in [1,z]^k$.
Thus, $t = g(\bar{t})$ is a word of the form $i_0\bar{s}$, where $\bar{s}$ consists of $i_1, \dots, i_{k}$ in some order, 
and $i_0$ does not occur anywhere in~$\bar{t}$.
Condition (i) is then immediate because of the choice of $i_0$. For condition (ii), we observe that 
$t' = g(\bar{t}')$ is a word of the form $i' \bar{s}'$, where $\bar{s}'$ consists of $i_0, i_2, \dots, i_{k}$ in some order, and $i'$ does not occur
 in any of the words in $\bar{t}'$. By the choice of $i'$, it is immediate that $t' \not \in \set{t_2, \dots, t_k}$. But the value $i_1$ (the first letter of $t_1$) 
occurs in $g(\bar{t})$ (which belongs to the tuple $\bar{t}'$), whence $i' \neq i_1$. It follows that $t' \neq t_1$, whence $t'$ is not in $\bar{t}$, as required.
\end{proof}

%% file: sections/appendix-upper-bounds.tex
%!TEX root = ../main.tex

\section{Appendix for Section~\ref{section:AF-upper-bounds}}

\subsection{Normal Form}
\newtheorem*{RestateTheorem_anf}{Lemma~\ref{lma:anf}}
\begin{RestateTheorem_anf}
	Let $\phi$ be a sentence of \AFv{\ell+1}, where $\ell \geq 2$. We can compute, in polynomial time, an \AFv{\ell+1}-formula 
	$\psi$ satisfiable over the same domains as $\phi$, of the form
	\begin{equation}
	\bigwedge_{i \in I} \forall \bx_{\ell} \exists x_{\ell+1}\, \gamma_i \wedge 
	\forall \bx_{\ell+1}\,
	\delta,
	\label{eq:anfRep}
	\end{equation}  
	where $I$ is a finite index set, and
	the formulas $\gamma_i$ and $\delta$ are quantifier-free.
\end{RestateTheorem_anf}

\begin{proof}
If the sentence $\phi$ is quantifier-free, then it is a formula of the propositional calculus, and the result is easily obtained by adding vacuous quantification. 
Otherwise, write $\phi_0 = \phi$, and  let $\theta\coloneqq Q x_k\, \chi$
be a subformula of $\phi$, where $Q \in \set{\forall, \exists}$, such that $\chi$ is quantifier-free. 
Writing $\bar{\exists}= \forall$ and $\bar{\forall} = \exists$,
let $p$ be a new predicate 
of arity $k$, let $\phi_1$ be the result of replacing $\theta$ in $\phi_0$ by the atom $p(\bx_k)$, 
and let  $\psi_1$ be the formula 
\begin{equation*}
\forall \bx_k Q x_{k+1}\big(p(\bx_{k}) \rightarrow \chi\big) \wedge
\forall \bx_k \bar{Q} x_{k+1}\big(\chi \rightarrow p(\bx_{k})\big).
\end{equation*}
It is immediate that $\phi_1 \wedge \psi_1 \models \phi_0$. Conversely, if $\fA \models \phi_0$, then we may expand $\fA$ 
to a model $\fA'$ of $\phi_1 \wedge \psi_1$ by taking $p^{\fA'}$ to be the set of $k$-tuples $\bar{a}$ such that $\fA \models \theta[\bar{a}]
$. Evidently, $\phi_1$ is a sentence of \AFv{\ell+1}. 
Processing $\phi_1$ in the same way, and proceeding similarly, we obtain a set of formulas $\phi_2, \dots, \phi_m$ and $\psi_2, \dots, \psi_m$, with  $\phi_m$ quantifier-free and
$\phi_0$ satisfiable over the same domains
as $\psi_1 \wedge \cdots \wedge \psi_m \wedge \phi_m$. Since $\phi_m$ is a sentence, it is a formula of the
propositional calculus. By moving $\phi_m$ inside any quantified formulas, re-indexing variables and re-ordering conjuncts, we obtain a formula
$\psi$ of the form~\eqref{eq:anfRep}.
%Until $|\sig(\psi)| < |I|$, we may replace each conjunct $\forall \bx_{\ell} \exists x_{\ell+1}\, \gamma_i$ ($i \in I$) by  
%$\forall \bx_{\ell} \exists x_{\ell+1}\, (\gamma_i \wedge p_i(\bx_\ell x_{\ell{+}1}))$ with $p_i$ a fresh $(\ell{+}1)$-ary predicate,
%without affecting the domains over which $\phi$ is satisfiable.
%Moreover, $\sig(\psi)$ contains no predicates of arity 0 (proposition letters), since their truth-values may be guessed.
\end{proof}

\subsection{Complexity of satisfiability for $\AFv{3}$}
\newtheorem*{RestateTheorem_simpleConnector}{Lemma~\ref{lma:simpleConnector}}
\begin{RestateTheorem_simpleConnector}
	If $\phi$ is a normal-form $\AFv{3}$-formula, $\fA \models \phi$ and $a \in A$, then  
	$\con^\fA[a]$ is compatible with $\phi$.
\end{RestateTheorem_simpleConnector}
\begin{proof}
	Suppose $b \in A$, and let $\con^\fA[b] = \omega$.
	For L$\exists_1$, consider $i \in I$. Since $\fA \models \phi$, there exists $c \in A$ such that
	$\fA \models \gamma_i[b,b,c]$.  Then $\eta = \tp^\fA[b,c]$ is as required. 
	For L$\exists_2$, suppose $\zeta^{-1} \in \omega$; thus,
	there exists $a \in A$ such that $\fA \models \zeta[a,b]$. 
	Now consider~$i \in I$. Since $\fA \models \phi$, there exists $c \in A$ such that
	$\fA \models \gamma_i[a,b,c]$. Then $\eta = \tp^\fA[b,c]$ is as required. The conditions L$\forall_1$ and
	L$\forall_2$ are established similarly. 
\end{proof}

\newtheorem*{RestateTheorem_simpleConnectorSet}{Lemma~\ref{lma:simpleConnectorSet}}
\begin{RestateTheorem_simpleConnectorSet}
	Let $\fA$ be a structure. Then $\Omega = \set{\con^\fA[a] \mid a \in A}$ is coherent.
\end{RestateTheorem_simpleConnectorSet}
\begin{proof}
	For G$\exists$, suppose $\omega \in \Omega$, and let $a \in A$ be such that $\con^\fA[a] = \omega$. 
	If $\zeta \in \omega$, there exists $b \in A$ such that $\tp^\fA[a,b] = \zeta$. Then $\omega' = \con^\fA[b]$ is
	as required. G$\forall$ is similar. 
\end{proof}

\newtheorem*{RestateTheorem_smallConnector}{Lemma~\ref{lma:smallConnector}}
\begin{RestateTheorem_smallConnector}
	Any satisfiable normal-form $\AFv{3}$-formula has a certificate $\Omega$ such that both $|\Omega|$ and $|\bigcup \Omega|$ are  $2^{O(\sizeof{\phi})}$.
\end{RestateTheorem_smallConnector}
\begin{proof}
	Suppose $\fA \models \phi$. Say that a word $\bar{a} \in A^m$ ($m \geq 1$) is {\em sensitive for $\phi$} if $\bar{a} = \bar{b}^g$ 
	for some word $\bar{b} \in A^k$
	with $k \leq 3$ and some $g \in \bA^m_k$ such that
	$\bar{x}_k^g$ occurs as the argument sequence of an atom in $\phi$. 
	Now define $\fA'$ by setting, for any $p$ of arity $m$ interpreted by $\fA$,
	$p^{\fA'} = \set{\bar{a} \mid \text{$\bar{a} \in p^\fA$ and $\bar{a}$ is sensitive}}$. In other words, all tuples except substitution instances of argument sequences occurring in $\phi$ are removed from the extensions of all predicates (of the appropriate arity). Clearly, $\fA' \models \phi$, since we
	have not modified any tuples which could make a difference to the evaluation of $\phi$. On the other hand, the number of 2-types
	realized in $\fA'$ is $2^{O(\sizeof{\phi})}$. Let 
	$\Omega' = \set{\con^{\fA'}[a] \mid a \in A}$. By Lemmas~\ref{lma:simpleConnector} and~\ref{lma:simpleConnectorSet}, $\Omega'$ is a certificate
	for $\phi$. We shall construct a certificate $\Omega \subseteq \Omega'$ satisfying the size bound of the lemma.
	Pick any $\omega \in \Omega'$, initialize $C\coloneqq \set{\omega}$, and $T\coloneqq \set{\zeta^{-1} \mid \zeta \in \omega}$. 
	We shall add connector-types to the set $C$ and 2-types to the set $T$ (all of them realized in $\fA'$), maintaining the invariant $T
	= \set{\zeta^{-1} \mid \text{$\zeta \in \omega$ for some $\omega \in C$}}$.
	We call a connector-type $\omega$ in
	$C$ {\em satisfied} if $\omega \subseteq T$. 
	Now execute the following procedure until $C$ contains no unsatisfied 
	connector-types.
	Pick some unsatisfied $\omega \in C$ and some $\zeta \in \omega \setminus T$. 
	By G$\exists$, there exists a connector-type $\omega' \in \Omega'$ such that $\zeta^{-1} \in \omega$. Set $C \coloneqq C \cup \set{\omega'}$ and $T \coloneqq T \cup \set{\zeta^{-1} \mid \zeta \in \omega'}$. These assignments maintain the invariant on $C$ and $T$. This process terminates after, say, $2^{O(\sizeof{\phi})}$ steps, since $|T|$ increases by at least one at each step, and $|C|$ by exactly 1. On termination, $C$ is globally coherent, and $|C| \leq 2^{O(\sizeof{\phi})}$. Set $\Omega$ to be the final value of~$C$.
\end{proof}

% \subsection{Verification of model size for $\AF^\ell$}

% \newtheorem*{RestateTheorem_upperBoundTh}{Theorem~\ref{theo:upperBound}}
% \begin{RestateTheorem_upperBoundTh}
% 	If $\phi$ is a satisfiable \AFv{\ell+1}-formula, with $\ell \geq 2$, then 
% 	$\phi$ is satisfied in a structure of size at most $\ft(\ell{-}1, O(\sizeof{\phi}))$.
% 	Hence the satisfiability problem for \AFv{\ell} is in $(\ell{-}2)$-\NExpTime{} for all $\ell \geq 3$, and the adjacent fragment is $\Tower$-complete.
% \end{RestateTheorem_upperBoundTh}